\documentclass{amsart}

\usepackage{amssymb,amsfonts}
\usepackage{graphicx}

\usepackage{hyperref}
\usepackage{float}

\newcommand\CO[2]{\raisebox{-.5\baselineskip}{$\overset{\textstyle #1\mathstrut}{#2}$}}
\newcommand\ROLO[2]{#1\{#2\}}
\newcommand{\p}{{\bf p}}

\newtheorem{theorem}{Theorem}[section]

\newtheorem{prop}[theorem]{Proposition}
\newtheorem{cor}[theorem]{Corollary}
\newtheorem*{schur}{Schur's Lemma}

\theoremstyle{definition}
\newtheorem{definition}[theorem]{Definition}
\newtheorem{example}[theorem]{Example}

\theoremstyle{remark}
\newtheorem{remark}[theorem]{Remark}
\newtheorem{proposal}[theorem]{Proposal}
\newtheorem{scenario}[theorem]{Scenario}
\newtheorem{notation}[theorem]{Notation}

\numberwithin{equation}{section}

\begin{document}

\title[Voting on Cyclic Orders]{Voting on Cyclic Orders, Group Theory, and Ballots}


%    author one information
\author[Crisman]{Karl-Dieter Crisman}
\address{Department of Mathematics and Computer Science, Gordon College}
\curraddr{}
\email{karl.crisman@gordon.edu}
\thanks{This work was partly supported by a Gordon College Provost's Summer Undergraduate Research Fellowship.}

\author[Holleran]{Abraham Holleran}
\address{Gordon College}

\author[Martin]{Micah Martin}
\address{Gordon College}

\author[Noonan]{Josephine Noonan}
\address{Gordon College}


\date{}

% USE?
% USE?

\begin{abstract}
A \emph{cyclic order} may be thought of informally as a way to seat people around a table, perhaps for a game of chance or for dinner.  Given a set of agents such as $\{A,B,C\}$, we can formalize this by defining a cyclic order as a permutation or linear order on this finite set, under the equivalence relation where $A\succ B\succ C$ is identified with both $B\succ C\succ A$ and $C\succ A\succ B$.  As with other collections of sets with some structure, we might want to aggregate preferences of a (possibly different) set of voters on the set of possible ways to choose a cyclic order.  

However, given the combinatorial explosion of the number of full rankings of cyclic orders, one may not wish to use the usual voting machinery.  This raises the question of what sort of ballots may be appropriate; a single cyclic order, a set of them, or some other ballot type?  Further, there is a natural action of the group of permutations on the set of agents.  A reasonable requirement for a choice procedure would be to respect this symmetry (the equivalent of neutrality in normal voting theory).

In this paper we will exploit the representation theory of the symmetric group to analyze several natural types of ballots for voting on cyclic orders, and points-based procedures using such ballots.  We provide a full characterization of such procedures for two quite different ballot types for $n=4$, along with the most important observations for $n=5$.
\end{abstract}

\maketitle

\section{Introduction}

\subsection{New outputs and new ballots}\label{subsection:newthings}

Two dominant paradigms for voting theory (going back before \cite{ArrowDifficulty} to names like Borda, Condorcet, Cusanus, and Llull) have long been the search for either a single winner or a full ranking of candidates.  Although these are evidently very useful models, even the fact that one needs definitions of terms like `resolute' and `social choice function' (both of which have to do with the inevitable occurrence of multiple winning outcomes in a choice procedure) indicates that other paradigms might be useful.

One very relevant area that has attracted a fair amount of recent attention is voting for committees -- to wit, fixed-size multiwinner elections.  Many people live in polities which allow for outcomes (and/or ballots) which are fixed-size subsets of the set of candidates, from school boards to national representatives.  Yet, as the recent survey \cite{FaliszewskiEtAlNewChallenge} of the issue in computational social choice declares, despite some initial work, this paradigm is still a relatively `new challenge' for the community.

Similarly, there has been an increasing amount of discussion about different types of ballots.  The most well-known alternate ballot is the simplest, coming from approval voting (see works from \cite{ApprovalVoting} to \cite{ApprovalHandbook}).  An approval ballot is one where an agent simply selects a subset of all possible candidates as `approved'.  This has been recognized as a ballot and not just a voting system, as indicated by its use in selecting committees in \cite{SatisfactionApproval}.  Other variant (but intuitive) ballot types include having three options (\cite{DisAndApproval}) as well as ballots from systems such as range voting or cumulative voting.  

Returning to committees, in \cite{RatliffPublicChoice} the case is largely made \emph{against} submitting a full ranking of all outcomes, whereupon \cite{RatliffSaariComplexities} instead restricts the ballot to picking one of a restricted subset of allowable committees.  In a somewhat different context, \cite{DavisOrrisonSu} asks for ballots which are \emph{subsets} of a committee.  Which ballot is best?  Suffice to say\footnote{More exotic ballot types include truncated rankings (such as in voting for several college football rankings in the United States) 
% Heisman, polls
and even more interesting options, such as a partially ordered set of the candidates in an election (\cite{ackerman-etal,CullinanEtAl,PiniEtAl,FaginEtAlPartial}).   % Possibly also Terzopolou?
In the more general aggregation problems that come up in data analysis (in everything from genetics to search engines), one could even have just binary relations as both ballot and output, as in the classic paper \cite{BarthelemyMonjardetMedian}.} that the door is just opening when it comes to what to do in voting \emph{and} balloting on less typical combinatorial objects appearing in voting theory.

\subsection{Some new scenarios}

This paper addresses both questions -- that of a new output type and that of which ballots to use -- for \emph{cyclic orders}.    Consider the following possible scenarios.

\begin{scenario}\label{scen:table}
Your large family is gathering for some important holiday (such as Thanksgiving, in the United States).  Everyone needs to be seated around a gigantic round table, but there are many things to keep in mind!  Here is a list of just some of \emph{your} priorities.
\begin{itemize}
\item Your leftist uncle must be kept far away from your rightist aunt.
\item Your uncle has two young children which should be seated very close to him -- ideally, adjacent.
\item Your grandmother is deaf in her right ear, so if she is seated next to an adult she likes to chat with, she must have that person to her \emph{left}.
\end{itemize}
And everyone else has their own priorities as well.  How do you come up with a seating everyone can live with, but without too much fuss?
\end{scenario}

\begin{scenario}\label{scen:courses}
A large faculty needs to vote on a rotation for topical course offerings that are given on a two- or three-year rotation.  Here are some sample (partial) preferences of faculty in the department.
\begin{itemize}
\item The $p$-adic Hodge theory course should be offered immediately after the basis $p$-adic analysis course, since they can build on each other.
\item Those courses should be offered far apart, so that there is always a $p$-adic course available.
\item The Computational Social Choice course should only be offered the semester immediately \emph{before} the Cooperative Game Theory course, so that most students in Cooperative Game Theory will be able to do advanced projects on $NP$-completeness.
\item The only person qualified to teach Computational Social Choice is also the same one as must teach Cooperative Game Theory, and she does not want to prepare both of them in the same academic or calendar year.
\end{itemize}
Even if your whole department votes on a rotation, what method should you choose to do so?
\end{scenario}

\begin{scenario}\label{scen:poker}
Poker is a gambling game of both luck and skill.  A long poker game can consist of many rounds, where the starting player shifts around a table after each round, so that in the long run everyone bids first about the same number of times.   In addition, it is well-known\footnote{See for example, \href{https://www.pokernews.com/strategy/10-hold-em-tips-08-25417.htm}{https://www.pokernews.com/strategy/10-hold-em-tips-08-25417.htm}.} that it is advantageous to be in bidding order \emph{after} a strong player, and to be placed immediately \emph{before} a weak player.  

Suppose that for a particular live-streaming poker championship, the audience of thousands is allowed to place bets\footnote{Although we do not endorse gambling as such in this paper, such meta-gambling certainly occurs.  One can place bets on the winner of the World Series of Poker final round courtesy of the owner of the tournament, Caesar's.} as well on who will win -- and is allowed to \emph{vote} on what order the players sit in!  Clearly an advantageous seating arrangement for your chosen winner will help your odds of winning.  How should the organizers of this tournament create ballots, and how should they tally them?
\end{scenario}

It should be clear that such scenarios could be multiplied indefinitely.  In each of them, the common factor is that there is some set of agents with preferences on some non-linear, `cyclic'  domain.  However, the most important aspects to agents seem to differ -- is it proximity, order, or some combination?  As above, submitting an entire ranking of all possible such orderings could be unwieldy, while simply using plurality on the set of all the orderings is too likely to give a massive tie.

\subsection{Cyclic Orders}

The intent of this paper is to examine voting on these cyclic orders, as well as to begin examining some alternate ballots of interest.  First, let us give a proper definition.

\begin{definition}\label{defn:cyclicorder}
Let $\mathcal{A}=\{A_1,A_2,\ldots ,A_n\}$ be a finite set and $S_n$ be the symmetric group.  Consider the set of permutations (or complete linear orders) $\mathcal{L}(\mathcal{A})$ of $\mathcal{A}$, and let $\sigma=(1234\cdots n)\in S_n$ be a natural $n$-cycle.  Then construct the equivalence relation where 
$$A_{i_1}\succ A_{i_2}\succ \cdots \succ A_{i_n}  \sim A_{i_{\sigma(1)}}\succ A_{i_{\sigma(2)}}\succ \cdots \succ A_{i_{\sigma(n)}} $$ 
for any element of $\mathcal{L}(A)$ (and hence also $A_{i_{\sigma^k(1)}}\succ A_{i_{\sigma^k(2)}}\succ \cdots \succ A_{i_{\sigma^k(n)}}$ for any $k$).  A (total) \emph{cyclic order} (sometimes \emph{cyclic ordering}) on $\mathcal{A}$ is an equivalence class of $\mathcal{L}(A)$ under this equivalence relation.  

We will occasionally use the notation $CO_n$ for this set for a given $n$, regardless of $\mathcal{A}$.  There are necessarily $n!/n=(n-1)!$ cyclic orders for a set of cardinality $n$.  
\end{definition}

\begin{notation}
Most of the time, it will be easier to think of a cyclic order as a single permutation in cycle notation, beginning with a common element.  If $\mathcal{A}=\{1,2,3,4,5\}$, then we might write $(13524)$ to represent the cyclic order where $1\succ 3\succ 5\succ 2\succ 4\succ 1$, where it is implied that $(35241)$, $(52413)$, and other analogous cyclic reorderings are equivalent.  Sometimes, to emphasize the cyclic nature, we may instead write $(135241)$.  However, we will usually use a set of the form $\mathcal{A}=\{A,B,C,D,\ldots\}$ because this is the traditional one to use in a voting context, though to be clear our methods will not be voting on them as candidates!
\end{notation}

\begin{remark}\label{remark:symmetricaction}
The above action is a right action of $S_n$.  But there is also an important natural \emph{left} action on the set of cyclic orders by the symmetric group $S_n$.  Given $\sigma \in S_n$, the action takes
$$A_{i_1}\succ A_{i_2}\succ \cdots \succ A_{i_n} \succ A_{i_1} \stackrel{\sigma}{\longrightarrow} A_{\sigma(i_1)}\succ A_{\sigma(i_2)}\succ \cdots \succ A_{\sigma(i_n)}\succ A_{\sigma(i_1)}$$
which one may think of as permuting the \emph{elements} of $\mathcal{A}$, rather than their \emph{placement} as in Definition \ref{defn:cyclicorder}.  For example, $(13)(A_1A_2A_3A_4)=(A_3A_2A_1A_4)\sim(A_1A_4A_3A_2)$.  Again, for convenience we may instead write this action as $(AB)(ABCD)=(ACDB)$, with $(AB)\in S_n$ but $(ABCD),(ACDB)\in CO_4$, when there should be no danger of confusion.
\end{remark}

\begin{remark}\label{remark:cyclicternary}
One can also define a cyclic order in terms of \emph{ternary} relations with certain properties.  While cyclic orders have been studied in terms of set properties and certain graphs (\cite{WoodallCyclicOrder} is a notable paper along these lines), they are almost always treated in the literature either as relating to \emph{partial} cyclic orders, or as a means to some other end.   See Section \ref{section:further} for some additional references. 
\end{remark}

We will consider several types of ballots related to cyclic orders.  In addition to `plurality' type ballots which ask for just the voter's favorite cyclic order (which will be addressed in Sections \ref{section:cyclic4} and \ref{section:cyclic5}), we will also consider ballots which request certain additional pieces of information from voters.   For instance, in Scenario \ref{scen:table} we could ask for ballots which allow a voter to select someone to placed immediately to their right as well as a person as far away from them as possible, or in Scenario \ref{scen:poker} one could submit a ballot placing optimal players to either side of the player you are actually betting on.  We will formalize notions like this as TRAD and ROLO ballots in Section \ref{section:4ballots}.

The remainder of the paper will unfold as follows:
\begin{itemize} 
\item We will start in Section \ref{section:cyclic4} by investigating procedures for $n=4$ spots in the cyclic order, given that the ballot is a single order.
\item Then we will recall some representation theory in Section \ref{section:repthry}, which will allow us to be far more systematic in our approach in the remainder of the paper.
\item Next, in Section \ref{section:4ballots} we will introduce and analyze various methods for $n=4$ which allow for more flexibility in balloting, but do not require a full ranking of all $6$ possible cyclic orders.
\item We will then use the same techniques in Section \ref{section:cyclic5} to examine the case $n=5$, once again given that the ballot is a single cyclic order.
\end{itemize}

\section{Cyclic orders for $n=4$}\label{section:cyclic4}

It can easily be checked that there are only two possible cyclic orders with $n=3$, so we begin by examining the situation where $n=4$.  In this case there are exactly six possible cyclic orders.  In Figure \ref{figure:cyclic4} we will set a particular order of these outcomes for the purposes of the paper, and include two types of representations which are convenient in different situations.

\begin{figure}[H]
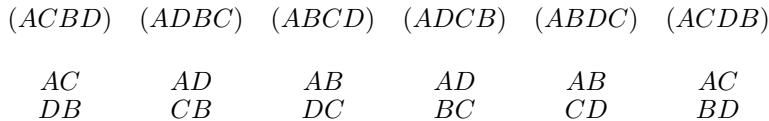

$$
\begin{array}{ccccccc}
(ACBD) & (ADBC) & (ABCD) & (ADCB) & (ABDC) & (ACDB)   \\
\\
\CO{AC}{DB} & \CO{AD}{CB} & \CO{AB}{DC} & \CO{AD}{BC} & \CO{AB}{CD} & \CO{AC}{BD}\\
\end{array}$$
\caption{Cyclic orders for $n=4$}
\label{figure:cyclic4}
\end{figure}

Note that the latter representations are meant to be read clockwise.  There are three pairs of orders, each of which is evidently the reversal of the other.

For convenience, in this section we may sometimes think of cyclic orders as seats at a table, as in Scenario \ref{scen:table}.  Think of some set of agents as voting on cyclic orders or seating arrangements, where in this section we allow each agent to select just one cyclic order (a `plurality' ballot, in some sense).  Then we formally extend the notion of a (necessarily anonymous) voting profile to this context.

\begin{definition}\label{defn:profile}
Given $n$ and an ordering of the $(n-1)!$ cyclic orders on a set of cardinality $n$, a \emph{profile} is a vector $\p\in \mathbb{Q}^{(n-1)!}$.  We write $\p[(ABCD)]$ to indicate the value of the profile for a particular cyclic order.  (If there are negative entries, especially if the sum of entries is zero, sometimes we call it a \emph{profile differential}.)
\end{definition}

\begin{remark}\label{remark:extend}
The $S_n$-action of Remark \ref{remark:symmetricaction} extends to profiles, via $\sigma(\p)[x]=\p[\sigma^{-1}x]$.  Writing elements of $S_4$ for convenience as $(AB)$, we see that if $\p=(2,1,1,0,0,0)^T$ then $(AB)\p = (1,2,0,0,0,1)^T$.
\end{remark}

Although there might be many ways to vote on cyclic orders, in this paper we confine ourselves to neutral points-based procedures\footnote{See \cite{ConitzerRognlieXia} and \cite{Zwicker} for inspiration and more general related notions.}, where neutrality should have the same connotation as it typically does in voting theory.

\begin{definition}\label{defn:points-based}
Suppose we have a \emph{ballot-scoring function}\footnote{Indeed, everything in this paper also works over the real numbers, but we restrict ourselves to the rationals for simplicity in certain situations.} $s:CO_n\times CO_n\to \mathbb{Q}$.  
Then a \emph{points-based voting rule} $f$, with ballot and outcome from the set of cyclic orders (for a given $n$), is a function from the set of all profiles on $CO_n$ to a nonempty subset of $CO_n$, such that $f(\p)$ is the (set of) cyclic order(s) $h$ which maximizes $\sum_{g\in CO_n} \p[g]s(g,h)$.  The rule $f$ is \emph{neutral} if $s$ is neutral, in the sense that for any $\sigma\in S_n$, $s(g,h)=s(\sigma(g),\sigma(h))$.
\end{definition}

Informally, the function $s$ represents how many points are allocated to the cyclic order $h$ for a single vote for the cyclic order $g$.  It is easier to envision this as applying a matrix $M$ (indexed by the elements of $CO_n$) with entries $M_{h,g}=s(g,h)$ to a profile $\p$, and then taking the argmax of the resulting vector, and we will do so in the sequel.    When, as is usual, there is no danger of ambiguity, we will very often refer to the outcome vector as the result of the voting rule, rather than the argmax.

\begin{example}\label{ex:rule210}
As an example, suppose we let $s$ be defined so that $s$ gives two points to a ranking for its appearance in a profile, but also awards one point to the reversed cyclic order.  In this case, 
$$M=\begin{pmatrix}
2 & 1 & 0 & 0 & 0 & 0 \\
1 & 2 & 0 & 0 & 0 & 0 \\
0 & 0 & 2 & 1 & 0 & 0 \\
0 & 0 & 1 & 2 & 0 & 0 \\
0 & 0 & 0 & 0 & 2 & 1 \\
0 & 0 & 0 & 0 & 1 & 2 \\
\end{pmatrix}\; .$$
With $\p=(2,1,0,0,0,1)^T$ we get $M\p=(5,4,0,0,1,2)^T$, so $(ACBD)$ is the most popular ballot in $\p$ as well as the winner under this rule.  

Note that this rule is also neutral; for example, the second column can be obtained from the first by applying $(AB)$ (or $(CD)$), and the third from the first by applying $(BC)$.
\end{example}

\begin{example}\label{ex:rule201}
The previous example's scoring function $s$ implied that a voter would consider a cyclic order to be related to its reversal, but not to others.  What if instead we wanted a model which implied the \emph{worst} thing that could happen for a voter for $(ACBD)$ was for the seating arrangement to be $(ADBC)$?  Then we might enjoy this matrix instead.
$$M=\begin{pmatrix}
2 & 0 & 1 & 1 & 1 & 1 \\
0 & 2 & 1 & 1 & 1 & 1 \\
1 & 1 & 2 & 0 & 1 & 1 \\
1 & 1 & 0 & 2 & 1 & 1 \\
1 & 1 & 1 & 1 & 2 & 0 \\
1 & 1 & 1 & 1 & 0 & 2 \\
\end{pmatrix}$$
With $\p=(2,1,0,0,0,1)^T$ we get $M\p=(5,3,4,4,3,5)^T$, with the perhaps surprising result that $(ACBD)$ and $(ACDB)$ share the victory.  Then again, three-quarters of the electorate wished for $A$ to sit to the right of $C$, so perhaps this is appropriate.
\end{example}

The matrices in question already seem to have a very high degree of symmetry.  This is not a coincidence.  

\begin{prop}\label{prop:4matrix}
Given the order of elements of $CO_4$ in Figure \ref{figure:cyclic4}, every neutral points-based voting rule for $n=4$ comes from a matrix as follows, with $a,b,c\in\mathbb{Q}$.
$$\begin{pmatrix}
a & b & c & c & c & c \\
b & a & c & c & c & c \\
c & c & a & b & c & c \\
c & c & b & a & c & c \\
c & c & c & c & a & b \\
c & c & c & c & b & a \\
\end{pmatrix}$$
\end{prop}

The proof follows by elementary methods; see the end of the section.  A corollary follows by simple computation.  

\begin{cor}\label{cor:4subspace}
With everything as in the previous proposition, every neutral points-based voting rule for $n=4$ can be described in terms of the linear algebra of $M$ as follows.
\begin{itemize}
\item Every profile vector in the subspace spanned by $\{(1,1,1,1,1,1)^T\}$ \emph{must} go to $a+b+4c$ times itself.
\item Every profile vector in the (two-dimensional) span of $\{(2,2,-1,-1,-1,-1)^T,$ $(-1,-1,2,2,-1,-1)^T, (-1,-1,-1,-1,2,2)^T\}$ \emph{must} go to $a+b-2c$ times itself.
\item Every profile vector in the subspace spanned by $\{(1,-1,0,0,0,0)^T,$ $(0,0,1,-1,0,0)^T,$ $(0,0,0,0,1,-1)^T\}$ \emph{must} go to $a-b$ times itself.
\end{itemize}
Finally, since these subspaces span $\mathbb{Q}^6$, any linear combination of the vectors above goes to the same linear combination of their images.
\end{cor}

Each of these subspaces also has an interpretation in terms of cyclic orders\footnote{All this should be strongly reminiscent of decompositions in \cite{ZwickerSpin,SaariBGOV,SaariJET3}, not to mention the treatments mentioned in the introduction to Section \ref{section:repthry}.}.  
\begin{itemize}
\item The first component is a trivial profile where all orders are equally likely.
\item The third space is comprised of linear combinations of profiles that strongly support one cyclic order and do not support its reversal.  
\item The second component may seem more mysterious, but is not if we recall Figure \ref{figure:cyclic4}.  Each of the \emph{three} vectors may be thought of as supporting one of the \emph{three} options of who (of $B,C,D$) is \emph{not} adjacent to $A$ in a given cyclic order.  Note also that these three vectors sum to zero, just like the Basic/Borda vectors in Saari's work.
\end{itemize}

\begin{example}\label{ex:rule210analyze}
Let's examine Example \ref{ex:rule210} from this standpoint.  We can write $\p=(2,1,0,0,0,1)^T$ as 
\begin{multline*}
\frac{2}{3}(1,1,1,1,1,1)^T + \frac{1}{3}(2,2,-1,-1,-1,-1)^T + \frac{-1}{6}(-1,-1,2,2,-1,-1)^T \\+ \frac{1}{2}(1,-1,0,0,0,0)^T+\frac{-1}{2}(0,0,0,0,1,-1)^T\; .
\end{multline*}
Since $a+b+4c=3$, $a+b-2c=3$, and $a-b=1$, we compute $f(\p)$ from
\begin{multline*}
2(1,1,1,1,1,1)^T + (2,2,-1,-1,-1,-1)^T + \frac{-1}{2}(-1,-1,2,2,-1,-1)^T \\+ \frac{1}{2}(1,-1,0,0,0,0)^T+\frac{-1}{2}(0,0,0,0,1,-1)^T
\end{multline*}
which is indeed $(5,4,0,0,1,2)^T$.
\end{example}

What is remarkable about this decomposition is that any neutral points-based rule can be fully described by three numbers, each of which simply dilate a given subspace of the profile space, and each of which (subspaces) has a natural interpretation.  One might say that the rule is really picking which \emph{spaces} to emphasize, rather than which cyclic orders.

As a result, we can analyze any method.  The one in Examples \ref{ex:rule210} and \ref{ex:rule210analyze} preserves non-adjacency more than it emphasizes the difference between a cyclic order and its reversal.  This isn't immediately intuitive, since the columns of the matrix are permutations of $(2,1,0,0,0,0)^T$, which would seem to definitely emphasize that reversals are not the same.  However, we can actually break \emph{this vector} into weighted pieces as well\footnote{Again, recall papers like \cite{SaariJET3} and \cite{OrrisonSymmetry}, as we will see below.}, so it is truly a combination of all emphases: $$(2,1,0,0,0,0)^T = \frac{1}{3}(1,1,1,1,1,1)^T+\frac{1}{2}(2,2,-1,-1,-1,-1)^T+\frac{1}{2}(1,-1,0,0,0,0)^T\; .$$

We can also create rules that have desired properties, just as with analyzing other voting procedures.  If we want a method that ignores full ties, we should ensure that $a+b+4c=0$.  Here is a more interesting one.

\begin{example}
Suppose we want a vote using a points-based rule that only takes into account the part of a profile that contrasts a cyclic order with its reversal.   One can think of this as requiring that the first and second subspaces in question span the kernel of the matrix.

So if we let $a+b+4c=0$ and $a+b-2c=0$, then by solving we see that $c=0$ and $a=-b$.  A sample matrix might be:
$$\begin{pmatrix}
1 & -1 & 0 & 0 & 0 & 0 \\
-1 & 1 & 0 & 0 & 0 & 0 \\
0 & 0 & 1 & -1 & 0 & 0 \\
0 & 0 & -1 & 1 & 0 & 0 \\
0 & 0 & 0 & 0 & 1 & -1 \\
0 & 0 & 0 & 0 & -1 & 1 \\
\end{pmatrix}$$
An observant reader will note this is just the matrix for \ref{ex:rule201} but with every entry lowered by one. 
\end{example}

It is evident from Corollary \ref{cor:4subspace} that there are only a few possible kernels and images (the spans of the subspaces in question).  This is a sign of the underlying symmetry.

\begin{proof}[Proof of Proposition \ref{prop:4matrix}]
Recall that our matrix has $M_{h,g}=s(g,h)$ for a scoring function $s$, which by hypothesis is neutral.  So our goal would be to show:
\begin{itemize}
\item That $s(g,g)$ is the same for all $g$.
\item That $s(g,g')$ is the same for all reversal pairs $g,g'$.
\item All other values of $s$ are the same as each other.
\end{itemize}
The first requirement is easy to see, because neutrality says $s(g,g)=s(\sigma(g),\sigma(g))$ and we certainly have $\sigma$ which change any cyclic order to another one.  Similarly, $s(g,h)=s(h,g)$, as it is easy to see by inspection that for any pair $g,h$ there is a $2$-cycle $\sigma_{g,h}$ such that $\sigma_{g,h}(g)=h$ and vice versa, so our matrix is symmetric.  

This shows that $s(g,g')=s(g',g)$ for \emph{each} set of reversal pairs $g,g'$.  To see that they are \emph{all} the same, note that $\sigma=(BC)$ sends the first reversal \emph{pair} $(ACBD),(ADBC)$ to the second pair, and $\sigma=(AC)$ does likewise to the third pair (and apply neutrality).  

Now consider $s(g,h)$ where $h\neq g,g'$.  Apply a $2$-cycle $\sigma$ such that $\sigma(g)=g'$ (such as $(AB)$ or $(CD)$ for $(ACBD)$).  Then $s(g,h)=s(g',\sigma(h))$, where $\sigma(h)\neq g,g',h$.  There are two possible choices for $\sigma$, which are disjoint, so their product (such as $(AB)(CD)$) sends $h$ to $h'$.  So if we do this twice we have that $s(g,h)=s(g,h')$ for any $g,h$ that are not a reversal pair!  By symmetry the same is true for $s(g,h)=s(g',h)$, and so all of the other entries are equal to each other as desired.
\end{proof}

\section{Representation Theory}\label{section:repthry}
There is one loose end in the previous section. Consider the subspaces in Corollary \ref{cor:4subspace}. Post-hoc it's easy to see they behave nicely, and with some thought about the structure of the cyclic orders one could come up with them `by hand'; still, it seems a little mysterious where they came from.  Likewise, although it's not a loose end, the proof of Proposition \ref{prop:4matrix} is quite ad-hoc, and one might despair of generalizing it.

In principle, it is nearly always \emph{possible} to analyze voting systems without any tools of higher power.  However, what is possible is not always doable in practice, as the sheer number of options demands more.  We choose, as this more powerful tool for analyzing both the systems and the spaces of profiles involved, the representation theory of the symmetric group.  This section will recall the information we need, with \cite{Sagan} as one of a number of handy references.

In doing so we follow the introduction of these techniques in mathematical social sciences, notably by \cite{KWAlgGames} in cooperative game theory and \cite{OrrisonSymmetry} in voting theory.  More powerful techniques do need some justification, of course.  One of the advantages of this approach is the ability to provide very general results; as an example, \cite{BarceloEtAl} and \cite{LeeThesis} do this with some of the questions asked in \cite{RatliffPublicChoice}.  Another advantage is to provide a unifying framework, as with \cite{OrrisonSymmetry} and \cite{CrismanPermuta} vis-a-vis the decompositions in the magisterial \cite{SaariStruct1,SaariStruct2}.  See \cite{CrismanOrrison} for a recent overview.

Although we began our analysis with the case where the ballot for a voting procedure is the same as a potential outcome, from Subsection \ref{subsection:newthings} and the comments after Remark \ref{remark:cyclicternary} we know our motivation is to be able to handle a wider variety of possible ballots, such as ones asking for partial information on adjacency or non-adjacency.  With that in mind, we might as well generalize Definition \ref{defn:points-based} of a points-based rule somewhat (to a generalized scoring rule, in the terminology of \cite{Zwicker}).  

\begin{definition}\label{defn:points-based-general}
Given $n$, suppose we have a \emph{ballot set} $B$, the set of \emph{profiles}\footnote{(Again, if there are negative entries, especially if the sum of entries is zero, sometimes we call it a \emph{profile differential}.)} $\mathbb{Q}^{\text{card}(B)}$ on $B$, and a \emph{ballot-scoring function} $s:B\times CO_n\to \mathbb{Q}$.  
Then a \emph{points-based voting rule} $f$, with \emph{ballot set} $B$ and \emph{outcome set} the set of cyclic orders (for a given $n$), is a function from the set of all profiles on $B$ to a nonempty subset of $CO_n$, such that $f(\p)$ is the (set of) cyclic order(s) $h$ which maximizes $\sum_{g\in B} \p[g]s(g,h)$.  

Further, assume that there is an action of $S_n$ on $B$ (extended to the set of profiles).  Then the rule $f$ is \emph{neutral} if $s$ is neutral, in the sense that for any $\sigma\in S_n$, $s(g,h)=s(\sigma(g),\sigma(h))$.
\end{definition}

As before, the function $s$ represents how many points are allocated to the cyclic order $h$ for a single vote for the ballot $g$.  The matrix $M$ we visualize now would be indexed by the elements of $CO_n$ in rows and elements of $B$ in columns, so that $M_{h,g}=s(g,h)$ as before; likewise, the system takes the argmax of the resulting vector $M\p$ to obtain the result for any profile $\p$.  Assuming that $B$ does have an $S_n$-action, we can also define $\sigma(\p)[x]=\p[\sigma^{-1}x]$ on the space of profiles, as in \ref{remark:extend}.

In that case, we can call the ballot space $\mathbb{Q}B$ and the outcome space $\mathbb{Q}CO_n$.  We make the following observations:
\begin{itemize}
\item By definition, our profile and outcome spaces have the $S_n$-action described, with a $\mathbb{Q}$-vector space structure compatible with $S_n$.   That is, both spaces have a $\mathbb{Q}S_n$-module structure.
\item Since the scoring functions are neutral ($s(g,h)=s(\sigma(g),\sigma(h))$), the $S_n$-action propagates from profiles to outcomes.  Given $\sigma\in S_n$ and $h\in CO_n$,
$$\sigma\left(\sum_{g\in B)}\p[g] s(g,h)\right)=\sum_{g\in B)}\p[(\sigma^{-1}(g)] s(\sigma^{-1}(g),\sigma^{-1}(h))\; ,$$ 
which is the same as the effect of $\sigma$ on the voting rule.  (This is the usual way neutrality would be defined.)
\item Hence\footnote{Abusing notation slightly as alluded to after \ref{defn:points-based}.} $$f:\mathbb{Q}B\to \mathbb{Q}CO_n$$ is a $\mathbb{Q}S_n$-module homomorphism.
\end{itemize}

Since we know $f$ is a $\mathbb{Q}S_n$-module homomorphism, we can use representation theory to analyze it, via decomposition of $\mathbb{Q}B$ and $\mathbb{Q}CO_n$ into irreducible submodules, and the following key results. 

\begin{prop}
Let $G$ be a group.  Any (finitely generated) $\mathbb{Q}G$-module has a decomposition as a direct sum of irreducible $\mathbb{Q}G$-modules which is unique up to order.
\end{prop}

\begin{schur}
Let $G$ be a group.  If $M$ and $N$ are irreducible $\mathbb{Q}G$-modules and $g:M\to N$ is a $\mathbb{Q}G$-module homomorphism, then either $g=0$ or $g$ is an isomorphism (\cite{Sagan} Theorem 1.6.5).  

Moreover, if $G=S_n$, then $g=0$ or $g$ is actually multiplication by an element of $\mathbb{Q}$ (\cite{Sagan} Corollary 1.6.8, along with \cite{JamesKerber} Theorem 2.1.12 that $\mathbb{Q}$ is a splitting field for $S_n$).
\end{schur}

\begin{prop}\label{prop:Specht}
Any irreducible representation of $S_n$ (irreducible $\mathbb{Q}S_n$-module) is isomorphic\footnote{This is true over both $\mathbb{Q}$ and $\mathbb{C}$ for the same reason as in the previous lemma.} to one of a finite set of (nonisomorphic) modules.  These modules are indexed by the partitions $\lambda$ of $n$ (denoted $\lambda \vdash n$), and are called the \emph{Specht modules} (\cite{Sagan} Theorem 2.4.6), denoted $S^\lambda$.
\end{prop}

\begin{example}\label{ex:3ballotspace}
The standard example of this in voting theory is when $n=3$ and the ballot space is $\mathcal{L}(\{A,B,C\})$.   In this case $(3),(2,1),(1,1,1)\vdash 3$ and $$\mathbb{Q}B\simeq \mathbb{Q}S_3 \simeq S^{(3)}\oplus S^{(2,1)^{\oplus^2}}\oplus S^{(1,1,1)}\; .$$ We can identify $S^{(3)}$ as the subspace of complete ties generated by $(1,1,1,1,1,1)^T$, and $S^{(1,1,1)}$ is the subspace generated by $(1,-1,1,-1,1,-1)^T$, where permutations coming from the cyclic order $(ABCA)$ have the opposite value of those coming from $(ACBA)$.   This is harnessed in \cite{SaariBGOV,SaariAllThree,ZwickerSpin}, by name in \cite{OrrisonSymmetry,CrismanPermuta}.
\end{example}

Note that in the previous example we had the fortune that the ballots in $B$ actually could be considered as elements of a group.  That does not happen with $CO_n$, unfortunately, but something \emph{nearly} as good is true which will allow us to decompose $\mathbb{Q}CO_n$ in practical cases.  We'll need a bit of additional terminology, but it will be worth it. 
First, we need to identify $\mathbb{Q}CO_n$ more precisely in group-theoretic terms.

\begin{definition}\label{defn:cyclicordercosets}
Let $H=\langle (123\cdots n)\rangle\leqslant S_n$.  We can restate \ref{defn:cyclicorder} by saying that $CO_n$ is the set of (left) cosets $S_n/H$.

Further, the $\mathbb{Q}S_n$-module $\mathbb{Q}CO_n$ is the \emph{permutation representation} of $S_n/H$.  This is defined by taking a vector space over $\mathbb{Q}$ with one basis vector $b_{gH}$ for each coset in $S_n/H$, and then letting $g'(b_{gH})=b_{g'gH}$ (and extending linearly).
\end{definition}

Finally, we need to recall the prime tool for \emph{computing} decomposition of representations.

\begin{definition}\label{defn:characterfacts}
Given an $\mathbb{Q}G$-module $V$, for each $g\in G$ there is a linear transformation on $V$ (as a finite-dimensional vector space) induced by applying $g$ to each basis element.  We can create the matrix $X(g)$ to represent this map.  We call its trace $\text{tr}(X(g))$ the \emph{character} of $V$ on $g\in G$, usually denoted $\chi_V$.
\end{definition}

\begin{prop}\label{prop:characterfacts}
The characters $\chi_\lambda$ of the irreducible modules $V_\lambda$ for $G$ are particularly nice, in that there is an inner product $(\chi,\chi')$ on the set of \emph{all} functions\footnote{In general from $G\to \mathbb{C}$, but over $S_n$ we may once again restrict to $\mathbb{Q}$.} $G\to \mathbb{Q}$ such that $(\chi_\lambda,\chi_V)$ is a nonnegative integer where there are exactly $(\chi_\lambda,\chi_V)$ `copies' of $V_\lambda$ in the decomposition of $V$.  (See \cite{Sagan} Section 1.9.)
\end{prop}

\begin{example}
In Example \ref{ex:3ballotspace}, we have $(\chi_{(2,1)},\chi_{\mathbb{Q}B})=2$, $(\chi_{(1,1,1)},\chi_{\mathbb{Q}B})=1$.
\end{example}

\begin{prop}
The permutation representation $\mathbb{Q}CO_n$ has character $\chi$ where\footnote{See e.g.~\cite{Sagan} Exercise 1.13.3b for this standard fact; see also formula (1.28) since this is actually an induced representation.} 
\begin{multline*}
\chi(g)=\text{ the number of fixed points of }g\text{ acting on }S_n/H\\
= \text{ the number of cyclic orders left unchanged by }g$$
\end{multline*}
\end{prop}

This is something we can compute!
Recall that characters are a \emph{class function}, that is to say the value of a character on a group element $g$ depends only on its conjugacy class in $G$.  Further recall that the classes of $S_n$ are precisely determined by their cycle decomposition (also indexed by the partitions of $n$).
\begin{theorem}\label{thm:cocharacter}
Let $D$ be the set of (positive) divisors $d$ of $n$, with divisor complement $e=\frac{n}{d}$.  Each divisor corresponds to a conjugacy class of $S_n$ of $C^e_d$. 
$$\chi(g)=
\begin{cases}
\frac{e!d^e}{n}\phi(d) & g\text{ of form }C^e_d\\
0 & \text{ otherwise }
\end{cases}$$
\end{theorem}

\begin{cor}\label{cor:cocharacterprime}
If $n=p$ is prime then $\chi(g)=
\begin{cases}
(n-1)! & g\text{ the identity }\\
n-1 & g\text{ an }n\text{-cycle}\\
0 & \text{ otherwise }
\end{cases}$
\end{cor}

\begin{proof}[Proof of Theorem \ref{thm:cocharacter}]
This is an easy 
exercise using standard induced character facts, but it is worth making an explicit calculation of the number of fixed points.

For a cyclic order $x$ to remain fixed under an element of $S_n$, all $A_i$ must be moved by the same amount, modulo $n$, in one of the underlying permutations representing $x$.  This immediately rules out any $g$ with more than one cycle type in the decomposition fixing any $x$, so assume $g$ has type $C^e_d$.

Without loss of generality let $g$ have the form $$(1\cdots \ell_{1,d})(2\cdots \ell_{2,d})\cdots (e\cdots \ell_{e,d})=g_1g_2\cdots g_e\; .$$  A typical cyclic order fixed by this has the form $(12\cdots e\ell_{1,2}\ell_{2,2}\cdots\ell_{e,d})$ so that each of the $d$-cycles are intertwined.  However, it is evident that other cyclic orders may be fixed, such as $(213\cdots e\ell_{2,2}\ell_{1,2}\ell_{3,2}\cdots \ell_{e,d})$.  In addition, the same cycle will be fixed if we replace a $d$-cycle $g_1=(1\cdots \ell_{1,d})$ by $(1\ell_{1,d}\ell_{1,d-1}\cdots \ell_{1,2})$, or any of the $\phi(d)$ powers of the cycle which preserve the cycle structure, so \emph{mutatis mutandis} there are $\phi(d)$ (cyclic!) orders of the elements of $g_1$ fixed by $g_1$ itself.  We now use a basic counting argument based on these observations.

Keeping the order of each $d$-cycle, we have $n=de$ choices of where to put the first element of the $d$-cycle $g_1$ among the elements of $x$.  In addition, we have $\phi(d)$ orderings of the elements of $g_1$ to choose from in constructing $x$.  For each of the remaining $d$-cycles $\{g_i\}_{i=2}^e$, we have $n-d(i-1)=d(e-i+1)$ positions of $x$ left to put the first element of $g_i$.  However, whichever of the $\phi(d)$ orderings of $g_1$ was chosen must now be chosen for $g_i$ as well, or $g$ will not \emph{fix} $x$, only the parts of $x$ corresponding to each $g_i$.  Finally, we have of course overcounted, because each of these $x$ is equivalent to $n$ other cyclic orders with the same properties.   Factoring out the powers of $d$ finishes the computation.
\end{proof}

We can reprove Proposition \ref{prop:4matrix} and Corollary \ref{cor:4subspace} from this point of view, which will prepare us for similar results in Section \ref{section:cyclic5}.  We first compute directly and with character tables for this proposition:

\begin{prop}\label{prop:co4decomp}
The character $\chi_{\mathbb{Q}CO_4}=\chi_4$ has values $\chi_4(e)=6$, \\$\chi_4((AB)(CD))=2$, $\chi_4((ABCD))=2$, and is zero otherwise.  Further, 
$$\mathbb{Q}CO_4\simeq S^{(4)}\oplus S^{(2,2)}\oplus S^{(2,1,1)}$$ of dimensions $1$, $2$, and $3$, respectively.
\end{prop}

The subspaces in Corollary \ref{cor:4subspace} we have already seen to be invariant under $S_4$, so by uniqueness of decomposition we can restate that result (using Schur's Lemma) as follows.

\begin{cor}\label{cor:co4decomp}
Any neutral points-based rule $f$ on $CO_4$ is determined by three scalars $t,u,v$ such that if $\p\in S^{(4)}\oplus S^{(2,2)}\oplus S^{(2,1,1)}$ is $\p=\p_1+\p_2+\p_3$, then $$f(\p)=t\p_1+u\p_2+v\p_3\; .$$
\end{cor}

  Further, we would have in Proposition \ref{prop:4matrix} that $t=a+b+4c$, $u=a+b-2c$, and $v=a-b$; solving for $a,b,c$ yields $$a = \frac{1}{6} \, t + \frac{1}{3} \, u + \frac{1}{2} \, v, b = \frac{1}{6} \, t + \frac{1}{3} \, u - \frac{1}{2} \, v, c = \frac{1}{6} \, t - \frac{1}{6} \, u$$ if one wanted to recreate the matrix in terms of these scalars instead.

Finally, we wish to recall one final very useful concept from \cite{OrrisonSymmetry}.  In any $\mathbb{Q}G$-module homomorphism, we have seen that the kernel is itself a $\mathbb{Q}G$-module.  However, for the purposes of voting theory this simply corresponds to the space of profiles that goes to a complete tie of zero points -- certainly a useful space, but we might be interested in what point totals are \emph{possible}, not just impossible.  

\begin{definition}\label{defn:effective}
Given a neutral points-based voting rule (or indeed a $\mathbb{Q}G$-module homomorphism) $f$ with (profile) domain $\mathbb{Q}B$, let the kernel be $K_f$.  Then call the orthogonal complement (using the usual inner product on $\mathbb{Q}^{\text{card}(B)}$) to $K_f$ the \emph{effective space of $f$}, and denote it by $E_f$.
\end{definition}

While any profile vector not in the kernel will go to some output vector other than a complete zero tie, any such profiles $\p$ can be written as $\p={\bf k} +{\bf e}$, where ${\bf k}\in K_f$ and ${\bf e}\in E_f$ are orthogonal.  Then $f(\p)=f({\bf e})$, so in some very real sense $E_f$ contains all the information there is to know.  Moreover, as we have seen in Corollary \ref{prop:co4decomp}/Corollary \ref{cor:4subspace}, any element of $E_f$ can be written as a sum of vectors which are then sent by scalar multiplication to vectors in the outcome space.  Identifying this space would therefore seem to be highly useful in constructing examples or proving theorems about any given $f$, or family of $f$.

\section{Meaningful ballots for $n=4$}\label{section:4ballots}

We now harness this symmetry to analyze more complex ballots.  Because the general results are pretty abstract, we will motivate them with some much more explicit computations.

\subsection{An in-depth example}

To start us off, consider Scenario \ref{scen:poker}.  There is a particularly nice setup here, because we can clearly identify a small subset of the many possible pieces of a cyclic order of greatest relevance to a voter.

First, each voter has a specific player $A_v$ he desires to win.  In addition, it seems reasonable that this voting public will be knowledgeable enough that each voter will be able to identify a strongest competitor $A_s$ and weakest competitor $A_w$ (neither the same as $A_v$).  In that case, allowing for a ballot that specifies the two adjacent players to a desired winner seems very reasonable (and certainly requires much less information than a full ranking of all cyclic orders).

\begin{definition}\label{defn:ROLOballot}
We define a \emph{ROLO ballot} to be a ballot that specifies a player, along with one to the Right Of and one to the Left Of this player.  We use the notation $\ROLO{A}{D,C}$, where we interpret this as desiring $A$ to have $D$ to its right and $C$ to its left, so that a generic cyclic order would look something like $\CO{DAC}{X\cdots Y}$.
\end{definition}

In the poker example, a good ballot to submit would be $\ROLO{A_v}{A_s,A_w}$.  However, it should be evident that the same ballot could be profitably used in many settings.  Note that there are $n(n-1)(n-2)$ ballots compared to $(n-1)!$ cyclic orders, which means that other than for $n\leq 6$ there are many more cyclic orders than possible ballots.   In this paper we will restrict ourselves to $n=4$, where there are already $24$ ballot types, so things are not dull.  

Let's apply Definition \ref{defn:points-based-general} to this situation.  Here, $B$ is the set of ROLO ballots, so $\mathbb{Q}B\simeq \mathbb{Q}^{24}$ as vector spaces.  Then once we devise a neutral scoring function $s:B\times CO_n\to \mathbb{Q}$, we will automatically have a new neutral points-based system to explore.

\begin{example}
Here is a system simple to describe informally, which is perhaps analogous to Example \ref{ex:rule210}.  Since an element of $B$ in this case completely describes a favorite cyclic order, given a ballot $b\in B$ we could give two points to that order, \emph{one} point to any other order that shares either the right or the left, and then zero points to everything else\footnote{Another way to think of this is that one point is awarded from $A\{D,C\}$ to any order with $AC$ in that order, and one point to any order with $DA$ in that order.}.   For example, a vote for $\ROLO{A}{D,C}$ would count two points to $(ACBD)$/$\CO{AC}{DB}$ and one point each for $(ACDB)$/$\CO{AC}{BD}$ and $(ABCD)$/$\CO{AB}{DC}$.  We will call this system ROLO(2,1).
\end{example}

Before analyzing this further, we will need to fix an order for the ballots.  Recall from Figure \ref{figure:cyclic4} the order given to $CO_4$ is\footnote{To keep the right and left in view easier one may recall $\CO{AC}{DB},\CO{AD}{CB},\CO{AB}{DC},\CO{AD}{BC},\CO{AB}{CD},\CO{AC}{BD}$.} $(ACBD)$, $(ADBC)$, $(ABCD)$, $(ADCB)$, $(ABDC)$, $(ACDB)$.  We will give the ROLO ballots the following order\footnote{The cognoscenti will later point out, when we see the group action explicitly, that the order could be even more symmetric, but redoing the various vectors and matrices involved is very error-prone, so we are sticking with this order!}, which has the four ballots corresponding to each cyclic order placed together.

\begin{figure}[H]
$$\begin{array}{ccccc}
\ROLO{A}{D,C} & \ROLO{B}{C,D} & \ROLO{D}{B,A} & \ROLO{C}{A,B}\\
\ROLO{C}{B,A} & \ROLO{D}{A,B} & \ROLO{A}{C,D} & \ROLO{B}{D,C}\\
\ROLO{A}{D,B} & \ROLO{C}{B,D} & \ROLO{B}{A,C} & \ROLO{D}{C,A}\\
\ROLO{B}{C,A} & \ROLO{D}{A,C} & \ROLO{C}{D,B} & \ROLO{A}{B,D}\\
\ROLO{D}{B,C} & \ROLO{A}{C,B} & \ROLO{B}{A,D} & \ROLO{C}{D,A}\\
\ROLO{C}{A,D} & \ROLO{B}{D,A} & \ROLO{D}{C,B} & \ROLO{A}{B,C}\\
\end{array}$$
\caption{ROLO ballots for $n=4$}
\label{figure:ROLO4}
\end{figure}

With this in mind, here is the full matrix for the points-based system ROLO(2,1).

\setcounter{MaxMatrixCols}{24}
\begin{figure}[H]
$$\begin{pmatrix}
%$$\left(\begin{smallmatrix}
2 & 2 & 2 & 2 & 0 & 0 & 0 & 0 & 1 & 0 & 0 & 1 & 1 & 0 & 1 & 0 & 1 & 0 & 1 & 0 & 1 & 0 & 0 & 1\\
0 & 0 & 0 & 0 & 2 & 2 & 2 & 2 & 0 & 1 & 1 & 0 & 0 & 1 & 0 & 1 & 0 & 1 & 0 & 1 & 0 & 1 & 1 & 0\\
1 & 0 & 1 & 0 & 1 & 0 & 0 & 1 & 2 & 2 & 2 & 2 & 0 & 0 & 0 & 0 & 0 & 1 & 1 & 0 & 1 & 0 & 1 & 0\\
0 & 1 & 0 & 1 & 0 & 1 & 1 & 0 & 0 & 0 & 0 & 0 & 2 & 2 & 2 & 2 & 1 & 0 & 0 & 1 & 0 & 1 & 0 & 1\\
0 & 1 & 1 & 0 & 1 & 0 & 1 & 0 & 1 & 0 & 1 & 0 & 0 & 1 & 1 & 0 & 2 & 2 & 2 & 2 & 0 & 0 & 0 & 0\\
1 & 0 & 0 & 1 & 0 & 1 & 0 & 1 & 0 & 1 & 0 & 1 & 1 & 0 & 0 & 1 & 0 & 0 & 0 & 0 & 2 & 2 & 2 & 2
%\end{smallmatrix}\right)$$
\end{pmatrix}$$
\caption{ROLO(2,1) matrix for $n=4$}
\label{figure:ROLO21matrix}
\end{figure}

The reader will immediately note a high degree of symmetry.  That symmetry comes from the comment, `the four ballots corresponding to each cyclic order,' and will be fully utilized momentarily. For now, let's see two representative results for ROLO(2,1).  First, a `paradox.'

\begin{example}\label{ex:ROLOparadox}
Suppose we have an electorate of about four thousand, distributed in the following profile:
%$$\left(\begin{smallmatrix}141, 141, 141, 141, 73, 313, 133, 253, 133, 159, 99, 193, 223, 9, 103, 129, 193, 159, 133, 219, 163, 9, 9, 163\end{smallmatrix}\right)$$
$$(141, 141, 141, 141, 73, 313, 133, 253, 133, 159, 99, 193, 223, 9, 103, 129, 193, 159, 133, 219, 163, 9, 9, 163)^T$$
About forty percent of the voters would like $(ACDB)$ or its reversal, and another third prefer $(ABCD)$ or its reversal.  However, under the ROLO(2,1) system, \emph{all four} of those options tie for last place, and $(ACBD)$ is the winning cyclic order by nearly one hundred scoring points!
\end{example}

To be sure, we do not know of any axioms for voting on cyclic orders that would cause this to be `paradoxical' in the usual voting sense.  At the same time, granted that we are currently only assuming neutrality and anonymity (not even a Pareto or unanimity axiom for the purposes of this paper!), this outcome seems troubling.  This is especially so since we went out of our way with ROLO ballots to not simply use a voting system equivalent to plurality.

On the positive side of the ledger, we have results like this.
\begin{example}\label{ex:ROLOtiespace}
Pick a cyclic order.  A profile differential over ROLO ballots with three voters for each ballot corresponding to that cyclic order, and negative three for each one corresponding to its reversal, yields an outcome vector of $24$ points for the order and $-24$ for its reversal.

This is unsurprising.  Moreover, we get a similar profile (differential), with an identical result, if we simply subtract the row of the matrix in Figure \ref{figure:ROLO21matrix} corresponding to the reversal of the chosen cyclic order from the row corresponding to the chosen one.

However, linearity immediately then implies that all profiles of forms like $$(1, 1, 1, 1, -1, -1, -1, -1, -1, 1, 1, -1, 1, -1, 1, -1, 1, -1, 1, -1, -1, 1, 1, -1)^T$$ go to a complete zero tie.  That is to say, ROLO(2,1) considers a profile (differential), preferring ballots coming from a given cyclic order exactly as much as it \emph{dislikes} ballots sharing only one adjacency with that same cyclic order, to be irrelevant to the scores.  
\end{example}

In some sense, this profile cancels out the points ROLO(2,1) gives to the two adjacencies of each possible ballot, as there are exactly the same number of each adjacency, positive and negative.  Again, without intuition about which profiles \emph{should} cancel out, this may or may not be controversial to assume as an axiom.  At the very least, though, it is a positive statement of something that is only obvious in retrospect about the rule.

\subsection{The regular representation}

Let's return to the symmetry of Figure \ref{figure:ROLO21matrix}.  What is it?  There are four elements of $S_4$ which fix each cyclic order, and we have already seen that these are the cosets in question defining cyclic orders.  That means that the action of $S_4$ on $B$ is isomorphic to that of $S_4$ \emph{on itself}, which is the definition (when linearized to vector spaces) of the regular representation.  Now let's apply Proposition \ref{prop:Specht} in this well-known case.

\begin{prop}\label{prop:rolodecomp}
The space $\mathbb{Q}B$ for ROLO ballots decomposes as $$S^{(4)}\oplus S^{(3,1)^{\oplus^3}}\oplus S^{(2,2)^{\oplus^2}}\oplus S^{(2,1,1)^{\oplus^3}}\oplus S^{(1,1,1,1)}\; .$$
\end{prop}

From Proposition \ref{prop:co4decomp} we then have a complete characterization of \emph{all} neutral points-based procedures on ROLO ballots!
\begin{prop}
The scoring part of any neutral points-based voting rule for $n=4$ from ROLO ballots to $CO_4$ is a $\mathbb{Q}S_4$-homomorphism
$$S^{(4)}\oplus S^{(3,1)^{\oplus^3}}\oplus S^{(2,2)^{\oplus^2}}\oplus S^{(2,1,1)^{\oplus^3}}\oplus S^{(1,1,1,1)} \to S^{(4)}\oplus S^{(2,2)}\oplus S^{(2,1,1)}$$
\end{prop}

\begin{cor}
Thus the kernel of \emph{any} such system is a superspace of the ten dimensions of $S^{(3,1)^{\oplus^3}}\oplus S^{(1,1,1,1)}$ direct summed with an eight-dimensional submodule isomorphic to $S^{(2,2)}\oplus S^{(2,1,1)^{\oplus^2}}$, while the effective space is isomorphic to an invariant subspace of $S^{(4)}\oplus S^{(2,2)}\oplus S^{(2,1,1)}$.  We may describe it fully with six scalars, which represent the scaling factors for the subspaces possibly not in the kernel $S^{(4)}\oplus S^{(2,2)^{\oplus^2}}\oplus S^{(2,1,1)^{\oplus^3}}$.
\end{cor}

There are several advantages to this point of view.  First off, it is much easier to find specific subspaces of interest once you know what you are looking for.  Secondly, we can not only describe how profiles will behave, but also more easily create `paradoxes', or at least unusual results.  

Before going in more depth, it will be helpful to examine the structure of Figures \ref{figure:ROLO4} and \ref{figure:ROLO21matrix} a little more closely.  As we pointed out, each ballot can be obtained from $\ROLO{A}{D,C}$ by a specific element of $S_4$ (considered as symmetries on the set $\{A,B,C,D\}$).  Less evidently, the matrix can also be obtained in the same way, by our neutrality requirement that $s(\sigma(b),\sigma(x))=\sigma(s(b,x))$.  This is evident in that each group of four columns corresponds to a coset of $\langle(ABCD)\rangle$ in $S_4$, which does not change the cyclic order that a given ballot favors, but \emph{does} cycle through the cyclic orders that share only the right-of or left-of characteristic.

We can now give far more explicit constructions both of any similar procedure (Proposition \ref{prop:QBmatrix}) \emph{and} of the space $\mathbb{Q}B$ (Figure \ref{figure:rolodecomp}).

\begin{prop}\label{prop:QBmatrix}
Any neutral points-based voting rule on a ballot of twenty-four options which is a transitive $S_4$-set to the set of cyclic orders must have a matrix of the following form:
$$\begin{pmatrix}
a & a & a & a & b & b & b & b & c & d & e & f & c & d & f & e & c & d & f & e & c & d & e & f\\
b & b & b & b & a & a & a & a & d & c & f & e & d & c & e & f & d & c & e & f & d & c & f & e\\
c & d & f & e & c & d & e & f & a & a & a & a & b & b & b & b & e & f & c & d & f & e & c & d\\
d & c & e & f & d & c & f & e & b & b & b & b & a & a & a & a & f & e & d & c & e & f & d & c\\
e & f & c & d & f & e & c & d & f & e & c & d & e & f & c & d & a & a & a & a & b & b & b & b\\
f & e & d & c & e & f & d & c & e & f & d & c & f & e & d & c & b & b & b & b & a & a & a & a
\end{pmatrix}$$
\end{prop}
\begin{proof}
Consider the permutations of the first column of weights which give the other columns.  The first four columns are generated by the $4$-cycle $(3546)$, and the next four by $(12)(56)$ in addition, forming a dihedral subgroup of $S_6$.  The next eight columns are achieved by a coset of that subgroup by $(135)(246)$ and the final eight columns are the coset generated by $[(135)(246)]^2=(153)(264)$.  The set of all these permutations is a subgroup of $S_6$ isomorphic to $S_4$, as can be easily checked.
\end{proof}

We can see that ROLO(2,1) has $a=2$ and $c=1=f$, otherwise zero.

In Figure \ref{figure:rolodecomp} we give the decomposition as well as brief voting-theoretic commentary.  There is \emph{not} a canonical way to distinguish among the infinitely many ways to write $S^{(2,2)^{\oplus^2}}$ or $S^{(2,1,1)^{\oplus^3}}$ as a direct sum of invariant subspaces, but our ordering of vectors shows one useful way, as we will see.  We used \cite{Sage} (in a non-automatic way) as well as computing characters to find this decomposition into orthogonal subspaces.  Bases are arranged in such a way to make it evident the symmetries are indeed the ones in question.

\begin{figure}[p]
Here is the decomposition of $\mathbb{Q}B$ as $S_4$-invariant subspaces.
\begin{itemize}
\item The trivial subspace $S^{(4)}=T$:
$$S^{(4)}=\text{span}(\left\{(1,1,1,1,1,1,1,1,1,1,1,1,1,1,1,1,1,1,1,1,1,1,1,1)^T\right\})$$
\item These four vectors span $S^{(2,2)^{\oplus^2}}$, and $\text{span}(\{v_1,v_2\})$ is an $S^{(2,2)}$ emphasizing one cyclic order/reversal pair and de-emphasizing the others.  
\begin{multline*}
v_1=(2, 2, 2, 2, 2, 2, 2, 2, -1, -1, -1, -1, -1, -1, -1, -1, -1, -1, -1, -1, -1, -1, -1, -1)^T,\\
v_2=(-1, -1, -1, -1, -1, -1, -1, -1, 2, 2, 2, 2, 2, 2, 2, 2, -1, -1, -1, -1, -1, -1, -1, -1)^T,\\
v_3=(2, 2, -2, -2, 2, 2, -2, -2, 1, 1, -1, -1, 1, 1, -1, -1, -1, -1, 1, 1, -1, -1, 1, 1)^T,\\
v_4=(1, 1, -1, -1, 1, 1, -1, -1, 2, 2, -2, -2, 2, 2, -2, -2, 1, 1, -1, -1, 1, 1, -1, -1)^T
\end{multline*}

\item These vectors span $S^{(2,1,1)^{\oplus^3}}$.   We will call each set $\{w_{i1},w_{i2},w_{i3}\}$ a \emph{triplet}; the spaces spanned by each set of three are both $S_4$-invariant and mutually orthogonal.  Note $\text{span}(\{w_{11},w_{12},w_{13}\})\simeq S^{(2,1,1)}$ is a subspace of vectors favoring a cyclic order strongly, and disfavoring its reversal strongly.
% 2 + 1 + 1 of each triple is the ROLO effective part of this subspace
% (2, 2, 2, 2, -2, -2, -2, -2, 1, -1, -1, 1, 1, -1, 1, -1, 1, -1, 1, -1, 1, -1, -1, 1)
% (1, -1, 1, -1, 1, -1, -1, 1, 2, 2, 2, 2, -2, -2, -2, -2, -1, 1, 1, -1, 1, -1, 1, -1)
% (-1, 1, 1, -1, 1, -1, 1, -1, 1, -1, 1, -1, -1, 1, 1, -1, 2, 2, 2, 2, -2, -2, -2, -2)
\begin{multline*}
w_{11}=(1,1,1,1,-1,-1,-1,-1,0,0,0,0,0,0,0,0,0,0,0,0,0,0,0,0)^T,\\
w_{12}=(0,0,0,0,0,0,0,0,1,1,1,1,-1,-1,-1,-1,0,0,0,0,0,0,0,0)^T,\\
w_{13}=(0,0,0,0,0,0,0,0,0,0,0,0,0,0,0,0,1,1,1,1,-1,-1,-1,-1)^T,\\
w_{21}=(0,0,0,0,0,0,0,0,1,-1,0,0,1,-1,0,0,1,-1,0,0,1,-1,0,0)^T,\\
w_{22}=(1,-1,0,0,1,-1,0,0,0,0,0,0,0,0,0,0,0,0,1,-1,0,0,1,-1)^T,\\
w_{23}=(0,0,1,-1,0,0,1,-1,0,0,1,-1,0,0,1,-1,0,0,0,0,0,0,0,0)^T,\\
w_{31}=(0,0,0,0,0,0,0,0,0,0,-1,1,0,0,1,-1,0,0,1,-1,0,0,-1,1)^T,\\
w_{32}=(0,0,1,-1,0,0,-1,1,0,0,0,0,0,0,0,0,-1,1,0,0,1,-1,0,0)^T,\\
w_{33}=(-1,1,0,0,1,-1,0,0,1,-1,0,0,-1,1,0,0,0,0,0,0,0,0,0,0)^T
\end{multline*}

\item The sign subspace $S^{(1,1,1,1)}$ simply gives a $\pm 1$ depending on whether the underlying permutation is even or odd; presumably this should be ignored by a system voting on cyclic orders, which will have equal amounts of each.
$$S^{(1,1,1,1)}=\text{span}(\left\{(1,1,-1,-1,1,1,-1,-1,-1,-1,1,1,-1,-1,1,1,1,1,-1,-1,1,1,-1,-1)^T\right\})$$
\item The $S^{(3,1)^{\oplus^3}}$ subspaces, spanned by the following triplets of vectors, all are in the kernel as well.
\begin{multline*}
u_{11}=(1, 1, -1, -1, -1, -1, 1, 1, 0, 0, 0, 0, 0, 0, 0, 0, 0, 0, 0, 0, 0, 0, 0, 0)^T,\\
u_{12}=(0, 0, 0, 0, 0, 0, 0, 0, -1, -1, 1, 1, 1, 1, -1, -1, 0, 0, 0, 0, 0, 0, 0, 0)^T,\\
u_{13}=(0, 0, 0, 0, 0, 0, 0, 0, 0, 0, 0, 0, 0, 0, 0, 0, 1, 1, -1, -1, -1, -1, 1, 1)^T,\\
u_{21}=(0, 0, 0, 0, 0, 0, 0, 0, -1, 1, 0, 0, -1, 1, 0, 0, 1, -1, 0, 0, 1, -1, 0, 0)^T,\\
u_{22}=(1, -1, 0, 0, 1, -1, 0, 0, 0, 0, 0, 0, 0, 0, 0, 0, 0, 0, -1, 1, 0, 0, -1, 1)^T,\\
u_{23}=(0, 0, -1, 1, 0, 0, -1, 1, 0, 0, 1, -1, 0, 0, 1, -1, 0, 0, 0, 0, 0, 0, 0, 0)^T,\\
u_{31}=(0, 0, 0, 0, 0, 0, 0, 0, 0, 0, 1, -1, 0, 0, -1, 1, 0, 0, 1, -1, 0, 0, -1, 1)^T,\\
u_{32}=(1, -1, 0, 0, -1, 1, 0, 0, 1, -1, 0, 0, -1, 1, 0, 0, 0, 0, 0, 0, 0, 0, 0, 0)^T,\\
u_{33}=(0, 0, -1, 1, 0, 0, 1, -1, 0, 0, 0, 0, 0, 0, 0, 0, -1, 1, 0, 0, 1, -1, 0, 0)^T
\end{multline*}
\end{itemize}
\caption{List of invariant subspaces for ROLO ballots}
\label{figure:rolodecomp}
\end{figure}

\begin{example}\label{ex:explainroloexamples}
Let's use Figure \ref{figure:rolodecomp} to explain our examples.  Recall that a generic\footnote{Of course, some, such as the zero matrix, will be trivial enough to have a bigger kernel and smaller effective space!} procedure will have a kernel and effective space component of $S^{(2,2)}$ and similarly for $S^{(2,1,1)}$ (twice as large of a kernel in the latter case), the effective parts of which will go directly to the corresponding space in $\mathbb{Q}CO_4$ in Corollary \ref{cor:4subspace}.  We can compute that the $S^{(2,1,1)}$ piece of the effective space of ROLO(2,1) is the space spanned by three vectors, each of which is the sum $2w_{1i}+w_{2i}+w_{3i}$.  (The kernel is the orthogonal complement to this, within that nine-dimensional subspace.)  This explains how Example \ref{ex:ROLOtiespace} was created.

Similarly, we can now see how Example \ref{ex:ROLOparadox} was created.  We started with a large piece of the $S^{(2,1,1)}$ effective space going to multiples of $(2,2,-1,-1,-1,-1)$ and added to it some large vectors in $S^{(2,2)}$ going to $(1,-1,0,0,0,0)$, so that $(ACBD)$ would be the winner.  Then, in order to obfuscate this reality, we found vectors in the kernel $S^{(2,1,1)}$ component (using linear algebra and Figure \ref{figure:rolodecomp}) which strongly favored $(ACBD)$ and \emph{subtracted those}.  Adding enough complete tie ballots to give positive entries finished things up.

For completeness, note that the span of $\{v_3,v_4\}$ in the $S^{(2,2)}$ bullet of Figure \ref{figure:rolodecomp} is the $S^{(2,2)}$ piece of the kernel of ROLO(2,1), while the span of $\{v_1,v_2\}$ is the effective space.

\end{example}

\subsection{Everything you can say}

Once we have these computations, one can make as general of statements as necessary for an application.  Although we hope people will consider voting/aggregating ballots related to cyclic orders, in some sense this is also an advertisement for the power of these methods.

As a concrete example, we might consider more general ROLO systems with the same ballot, such as ROLO($x$,1) where $a=x$ rather than $a=2$ as before.  This is a system that values more (or less) having exactly the `favorite' cyclic order over other cyclic orders compatible with the adjacencies in a ROLO ballot.

\begin{prop}\label{prop:rolox1spaces}
The effective and kernel $S^{(2,2)}$ subspaces for all ROLO($x$,1) procedures (except when $x=1$) are the same as in Example \ref{ex:explainroloexamples}.    The $S^{(2,1,1)}$ effective space for ROLO($x$,1) is spanned by the vectors $xw_{1i}+w_{2i}+w_{3i}$. 
\end{prop}

\begin{proof}
One can directly compute by a suitably modified matrix that the $S^{(2,2)}$ kernel for ROLO(2,1) also goes to zero for ROLO($x$,1), and likewise that the other two $S^{(2,2)}$ basis vectors (as in the statement) go where indicated.  Schur's Lemma and the orthogonality of spaces spanned by $\{v_1,v_2\}$ and $\{v_3,v_4\}$, respectively, in Figure \ref{figure:rolodecomp} complete the proof of the first statement.  When $x=1$ the kernel $S^{(2,2)}$ component is the whole four-dimensional subspace.  The verification of the second statement is very similar to Example \ref{ex:explainroloexamples}.  
\end{proof}

\begin{prop}\label{prop:rolox1spaces2}
We can obtain more explicit results with raw computation.
\begin{itemize}
\item 
Vectors fully supporting one cyclic order and its reversal over against a different reversal pair, such as
$$(1,1,1,1,1,1,1,1,-1,-1,-1,-1,-1,-1,-1,-1,0,0,0,0,0,0,0,0)^T$$
are sent to outcomes of the form $4(x-1,x-1,1-x,1-x,0,0)$.  (Verify by adding the $x$ and $1$ points for each cyclic order, without using the matrix.)
\item 
As in Example \ref{ex:ROLOtiespace}, in the matrix for ROLO($x$,1) we subtract the second row from the first (and likewise with the others) to get a basis.   The orthogonal complement to this is the kernel.  
\item 
On the other hand, vectors fully supporting one cyclic order over against its own reversal, such as
$$(0,0,0,0,0,0,0,0,0,0,0,0,0,0,0,0,1,1,1,1,-1,-1,-1,-1)\;, $$
are sent to outcomes of the form $4(0,0,0,0,x,-x)$ fully supporting in the same way.  Note that when $x=0$ these are in the kernel, so they are clearly \emph{not} the effective space.
\end{itemize}
\end{prop}

More generally, we have some much more powerful statements, applicable to \emph{any neutral points-based rule} for this ballot space along the lines of Proposition \ref{prop:QBmatrix}.  Call the matrix in that Proposition $M$, the rule $f_M$, and recall the constants $a$ through $f$.

\begin{theorem}
Clearly the image of $T$ is $4(a+b+c+d+e+f)(1,1,1,1,1,1)$, so $T$ is in the kernel precisely when the weighting vector is sum-zero.  

The $S^{(2,2)}$ effective space is the span of the following two vectors\footnote{Because of the symmetries involved, it does not seem possible to find a basis for the $S^{(2,2)^{\oplus^2}}$ space that both decomposes into two \emph{orthogonal} pieces, but also looks natural for the general case.}:
\begin{multline*}
(2  a + 2  b - c - d - e - f,2  a + 2  b - c - d - e - f,2  a + 2  b - c - d - e - f,2  a + 2  b - c - d - e - f,\\
2  a + 2  b - c - d - e - f,2  a + 2  b - c - d - e - f,2  a + 2  b - c - d - e - f,2  a + 2  b - c - d - e - f,\\
-a - b + 2  c + 2  d - e - f,-a - b + 2  c + 2  d - e - f,-a - b - c - d + 2  e + 2  f,-a - b - c - d + 2  e + 2  f,\\
-a - b + 2  c + 2  d - e - f,-a - b + 2  c + 2  d - e - f,-a - b - c - d + 2  e + 2  f,-a - b - c - d + 2  e + 2  f,\\
-a - b + 2  c + 2  d - e - f,-a - b + 2  c + 2  d - e - f,-a - b - c - d + 2  e + 2  f,-a - b - c - d + 2  e + 2  f,\\
-a - b + 2  c + 2  d - e - f,-a - b + 2  c + 2  d - e - f,-a - b - c - d + 2  e + 2  f,-a - b - c - d + 2  e + 2  f)\\
(-a - b + 2  c + 2  d - e - f,-a - b + 2  c + 2  d - e - f,-a - b - c - d + 2  e + 2  f,-a - b - c - d + 2  e + 2  f,\\
-a - b + 2  c + 2  d - e - f,-a - b + 2  c + 2  d - e - f,-a - b - c - d + 2  e + 2  f,-a - b - c - d + 2  e + 2  f,\\
2  a + 2  b - c - d - e - f,2  a + 2  b - c - d - e - f,2  a + 2  b - c - d - e - f,2  a + 2  b - c - d - e - f,\\
2  a + 2  b - c - d - e - f,2  a + 2  b - c - d - e - f,2  a + 2  b - c - d - e - f,2  a + 2  b - c - d - e - f,\\
-a - b - c - d + 2  e + 2  f,-a - b - c - d + 2  e + 2  f,-a - b + 2  c + 2  d - e - f,-a - b + 2  c + 2  d - e - f,\\
-a - b - c - d + 2  e + 2  f,-a - b - c - d + 2  e + 2  f,-a - b + 2  c + 2  d - e - f,-a - b + 2  c + 2 \, d - e - f)
\end{multline*}
The image is 
\begin{multline*}
4 \, a^{2} + 8 \, a b + 4 \, b^{2} - 4 \, a c - 4 \, b c + 4 \, c^{2} - 4 \, a d - 4 \, b d + 8 \, c d + 4 \, d^{2} \\
- 4 \, a e - 4 \, b e - 4 \, c e - 4 \, d e + 4 \, e^{2} - 4 \, a f - 4 \, b f - 4 \, c f - 4 \, d f + 8 \, e f + 4 \, f^{2}
\end{multline*}
times the second space of Corollary \ref{cor:4subspace}.

Finally, the $S^{(2,1,1)}$ effective space for $f_M$ has a basis comprised of the weighted vector sums of corresponding vectors in each `triplet', namely $$(a-b)w_{1i}+(c-d)w_{2i}+(f-e)w_{3i}$$
with image $4(a-b)^2+4(c-d)^2+4(f-e)^2$ times the third space of Corollary \ref{cor:4subspace}.
\end{theorem}

\begin{proof}
In some sense, this is all banal computation with the matrix $M$.  A particularly efficient way to observe these is to note that we desire $M\p = k{\bf v}$ for the various vectors ${\bf v}$ in the bases of Corollary \ref{cor:4subspace}.  Then consider the real symmetric (and hence orthogonally diagonalizable) matrix $MM^T$; an eigenvector ${\bf v}$ of this yields a solution $M^T{\bf v}=\p$ to our problem, and the eigenvalues coming from $MM^T$ are the multiplicative factors above.
\end{proof}

We can apply this result easily, if desired.  The following example extends Proposition \ref{prop:rolox1spaces2}:
\begin{example}
The basis vector in the $S^{(2,1,1)}$ effective space for ROLO($x$,1) supporting a particular cyclic order goes to $4x^2+8$ times\footnote{Although the algebra is linear, this shows that the way linear procedures \emph{vary} need not be linear.} the outcome supporting an order and not its reversal.  This recovers the result for ROLO(2,1) in Example \ref{ex:ROLOtiespace}, as $4\cdot 2^2+8=24$.
\end{example}

\subsection{Even more ballots}

Importantly, there is no reason to presume that the ballots in question are actually ROLO ballots.  To use this analysis, \emph{we only need a set of twenty-four ballots which are a transitive $S_4$-set.}  As a perhaps less intuitive example, imagine we are in the setting of Scenario \ref{scen:table}.  Suppose each voter can submit a ballot that consists not of a preferred cyclic order or adjacency, but two people who should definitely be \emph{far} apart from each other (which immediately gives the other pair which are opposite), and a pair who should be adjacent in a particular direction (compatible with the other information, of course).  

For example, there are four such ballots which would define $(ACBD)/\CO{AC}{DB}$, where $XY-ZW$ means $X$ and $Y$ are opposite and $Z$ is to the right of $W$: 
$$AB-DA \quad AB-CB \quad AB-BD\quad AB-AC$$
We may call such ballots TRAD\footnote{To the Right, Across Diagonally.} ballots, and a system which gives two points to a cyclic order for fulfilling both conditions, one for fulfilling one, might be called TRAD(2,1).  Indeed, we obtain this from the matrix in Proposition \ref{prop:QBmatrix} when $a=2$, $b=1=c$, and zero otherwise.  Here are two sample results, analogous to earlier ones for ROLO, and using the same techniques.

\begin{example}
In TRAD(2,1), one emphasizes the diagonal a fair bit.  Nonetheless, by adding elements from the effective $S^{(2,1,1)}$ subspace of TRAD(2,1) to the $S^{(2,2)}$ part of the \emph{kernel} of TRAD(2,1), we can create the following profile, which strongly emphasizes cyclic orders where $A$ is across from $B$, and somewhat less emphasizes those where $A$ is across from $D$.
$$(218, 198, 128, 128, 218, 198, 128, 128, 26, 26, 186, 186, 6, 6, 166, 166, 0, 0, 250, 230, 0, 0, 250, 230)$$
%$$\left(\begin{smallmatrix}218, 198, 128, 128, 218, 198, 128, 128, 26, 26, 186, 186, 6, 6, 166, 166, 0, 0, 250, 230, 0, 0, 250, 230\end{smallmatrix}\right)$$
Applying the appropriate matrix yields $(ABCD)$ as the winning cyclic order, which not very many people wanted!
\end{example}

\begin{example}
The effective space $S^{(2,1,1)}$ component for TRAD(2,1) is spanned by vectors of the form
$$(1, 1, 1, 1, -1, -1, -1, -1, 1, -1, 0, 0, 1, -1, 0, 0, 1, -1, 0, 0, 1, -1, 0, 0)$$
which are the sums $w_{1i}+w_{2i}$ in the $S^{(2,1,1)}$ in Figure \ref{figure:rolodecomp}.

Moreover, this effective space is in the kernel of ROLO(-1,1)!  In some sense, a rule which wants to promote either of two compatible directional adjacencies, but \emph{disadvantage} a cyclic order which has both, must a fortiori ignore a large part of diagonal information.
\end{example}

Although we will not pursue this here, the same analysis could be done for other ballot spaces.  An relatively easy one would be the set of ROLO ballots where $\ROLO{A}{CD}$ is identified with $\ROLO{A}{DC}$, since we could use our previous analysis -- but not directly, since we can't just have `half as big' of everything!  Doing this is a good exercise for anyone interested in using these techniques -- either by computing a character as in Theorem \ref{thm:cocharacter}, or by creating matrices.  Note that the outcome space is likely to be isomorphic to a quotient of $\mathbb{Q}CO_4$, since reversals can no longer be distinguished!

\section{The case of five}\label{section:cyclic5}

In this paper we do not attempt the characterizations for all $n$ which often appear in the literature.  This is largely because (as we have mentioned above) we don't have axiomatics for which subspaces are of most importance; contrast the situation with the $S^{(n-1,1)}$ and $S^{(n-2,1,1)}$ components in other related settings, such as voting on full orderings, or values for cooperative games.  The complexity of the character computations also grow substantially with $n$ and the number of partitions of $n$.

To give a taste of what lies in store, we end with a set of characterizations of neutral points-based procedures where voters pick a single cyclic order, as in Definition \ref{defn:points-based} with $n=5$; that is, where $B=CO_5$.   The earlier-mentioned combinatorial explosion will already be evident, but so will intriguing consistency with our $n=4$ observations.

\subsection{The decomposition}

First, Corollary \ref{cor:cocharacterprime} gives us the character of $\mathbb{Q}CO_5$, which we can then apply Proposition \ref{prop:characterfacts} to in order to find\footnote{The interactive calculator at \href{https://www.jgibson.id.au/articles/characters/}{https://www.jgibson.id.au/articles/characters/} is helpful for doing these computations, as well as being remarkably addictive.} the decomposition of $\mathbb{Q}CO_5$ in terms of irreducible $\mathbb{Q}S_5$-modules.  

\begin{prop}
We can decompose $\mathbb{Q}CO_5$ as a sum of irreducible $\mathbb{Q}S_5$-modules as follows:
$$\mathbb{Q}CO_5 \simeq S^{(5)}\oplus S^{(3,2)}\oplus S^{(3,1,1)^{\oplus^2}}\oplus S^{(2,2,1)}\oplus S^{(1,1,1,1,1)}$$
where these modules have dimensions $1$, $5$, $2\cdot 6$, $5$, and $1$, respectively.
\end{prop}
\begin{proof}
This is a standard calculation, so we just give $S^{(3,1,1)}$ as an example.  The only nonzero values of $\chi$ are for the identity and for the conjugacy class of $5$-cycles, and the values of $\chi_{(3,1,1)}$ for those are $6$ and $1$.  There is one element in the identity's class, and $24$ in the $5$-cycle class, so we compute: $$24\cdot 1\cdot 6 + 4\cdot 1\cdot 24 = 240=2\cdot 120$$ which verifies this irreducible has multiplicity two in the decomposition.
\end{proof}

\begin{cor}
Any neutral points-based rule on $CO_5$ where ballots are a single cyclic order only depends on eight parameters.
\end{cor}
\begin{proof}
We have one parameter each by Schur's Lemma for $S^{(5)}$, $S^{(3,2)}$, $S^{(2,2,1)}$, and $S^{(1,1,1,1,1)}$.  Further, each component of $S^{(3,1,1)^{\oplus^2}}$ (in any particular way to write it as a direct sum) needs two parameters to describe what linear combination of vectors in the (same two) components of $S^{(3,1,1)^{\oplus^2}}$ it goes to.
\end{proof}

This means that even though there are many more possible irreducibles for $n=5$ than $n=4$, the degrees of freedom are not much more than for the ROLO ballots for $n=4$.

More interesting is identifying which concrete subspaces actually correspond to the irreducibles, and whether these identify new voting-theoretic ideas.  For this purpose, we fix the following ordering\footnote{As with the ROLO ballots, there are several other possible useful orderings, some even more symmetric, but this one works well for our purpose.} of $CO_5$.  Each cyclic order is paired with its reversal.

\begin{figure}[H]
$$\begin{array}{cccccc}
(ABCDE) & (AEDCB) & (ABCED) & (ADECB) & (ABDCE) & (AECDB) \\
(ABDEC) & (ACEDB) & (ABECD) & (ADCEB) & (ABEDC) & (ACDEB) \\
(ACBDE) & (AEDBC) & (ACDBE) & (AEBDC) & (ACEBD) & (ADBEC) \\
(ADBCE) & (AECBD) & (AEBCD) & (ADCBE) & (ACBED) & (ADEBC)
\end{array}$$
\caption{Cyclic orders for $n=5$}
\label{figure:CO5}
\end{figure}

In Figure \ref{figure:co5decomp} we have the full decomposition of the vector space $\mathbb{Q}^{24}$ with respect to the $S_5$ action.  Some subspaces are easier to identify than others!  For instance, $S^{(5)}$ will clearly be the space spanned by $(1,1,\ldots ,1)^T$.  Since when $n=5$ any $5$-cycle is already an even permutation, we can likewise ignore the `cyclic' nature of the cyclic orders and see that the sign representation $S^{(1,1,1,1,1)}$ must simply have $\pm 1$ as according to whether the underlying permutation is even or odd.  The other irreducibles need some additional ingenuity (and computational assistance) to find, but once found then it is once again a standard character computation to verify they are correct.  

\begin{example}
As an example, consider the $S^{(3,1,1)^{\oplus^2}}$ component in Figure \ref{figure:co5decomp} and compare to the ordering in Figure \ref{figure:CO5}.  It might be surprising that the space generated by all profiles (or differentials) which support a given cyclic order and oppose its reversal would have only one type of irreducible, one type of symmetry -- and, despite the same thing happening in Corollary \ref{cor:4subspace}, this does not generalize even to $n=6$ or $n=7$.  

Yet the computation is clear.  The space is clearly completely fixed under the identity.  A given $5$-cycle will fix \emph{two} of the reversal pairs; for instance, $(ABCDE)$ fixes the pairs $(ABCDE),(AEDCB)$ and $(ACEBD),(ADBEC)$.  Finally, while most other cycle classes obviously cannot fix anything, a product of two-cycles, such as $(BE)(CD)$ can fix a pair, but swapping the places of a cyclic order and its reversal.  A given element does so with four of the pairs\footnote{For this element, it would be the pairs $(ABCDE),(AEDCB)$,  $(ACEBD),(ADBEC)$, $(ABDCE),(AECDB)$, and $(ACBED),(ADEBC)$ which have $B,E$ or $C,D$ surrounding $A$ in the cyclic order.}, so the trace of the action on this subspace for such an element is $-4$.  This is exactly twice the character for $S^{(3,1,1)}$ for all classes, so the space is isomorphic (not canonically, to be sure) to $S^{(3,1,1)}\oplus S^{(3,1,1)}$.
\end{example}

\begin{figure}[p]
Here is the decomposition $\mathbb{Q}CO_5$ as $S_5$-invariant subspaces.  It is helpful to think of the $24$ entries in each vector as corresponding to twelve sets of reversal pairs.
\begin{itemize}
\item The trivial subspace $S^{(5)}=T$:
$$S^{(5)}=\text{span}\left\{(1,1,1,1,1,1,1,1,1,1,1,1,1,1,1,1,1,1,1,1,1,1,1,1)^T\right\}$$

\item The sign subspace $S^{(1,1,1,1,1)}$ gives $\pm 1$ as outlined above.
$$S^{(1,1,1,1)}=\text{span}\left\{(1,1,-1,-1,-1,-1,1,1,1,1,-1,-1,-1,-1,1,1,-1,-1,1,1,-1,-1,1,1)^T\right\}$$ 

\item We then have two \emph{non}-isomorphic, but similar, subspaces.  $S^{(3,2)}$ is the span of the $\{y_i\}$ and $S^{(2,2,1)}$ is the span of the $z_i$.
\begin{multline*}
y_1=(5,5,-1,-1,-1,-1,-1,-1,-1,-1,-1,-1,-1,-1,-1,-1,5,5,-1,-1,-1,-1,-1,-1)^T,\\
y_2=(-1,-1,5,5,-1,-1,-1,-1,-1,-1,-1,-1,-1,-1,-1,-1,5,5,-1,-1,-1,-1,-1,-1)^T,\\
y_3=(-1,-1,-1,-1,5,5,-1,-1,-1,-1,-1,-1,-1,-1,-1,-1,-1,-1,-1,-1,-1,-1,5,5)^T,\\
y_4=(-1,-1,-1,-1,-1,-1,5,5,-1,-1,-1,-1,-1,-1,-1,-1,-1,-1,-1,-1,5,5,-1,-1)^T,\\
y_5=(-1,-1,-1,-1,-1,-1,-1,-1,5,5,-1,-1,5,5,-1,-1,-1,-1,-1,-1,-1,-1,-1,-1)^T\\
%(-1,-1,-1,-1,-1,-1,-1,-1,-1,-1,5,5,-1,-1,-1,-1,-1,-1,5,5,-1,-1,-1,-1)^T\}
%\end{multline*}
%\begin{multline*}
z_1=(5,5, % 1
1,1, % 2
1,1, % 3
-1,-1, % 4
-1,-1, % 5
1,1, % 6
1,1, % 7
-1,-1, % 8
-5,-5, % 9
-1,-1, % 10
1,1, % 11
-1,-1)^T,\\ % 12
z_2=(1,1,5,5,-1,-1,1,1,1,1,-1,-1,-1,-1,-5,-5,-1,-1,1,1,-1,-1,1,1)^T,\\ % do all swaps (DE)
z_3=(1,1,-1,-1,5,5,1,1,1,1,-1,-1,-1,-1,1,1,-1,-1,1,1,-1,-1,-5,-5)^T,\\ % do all swaps (CD)
z_4=(-1,-1,1,1,1,1,-1,-1,5,5,1,1,-5,-5,-1,-1,1,1,-1,-1,1,1,-1,-1)^T,\\ % do all swaps (BC) and also -
z_5=(1,1,-1,-1,-1,-1,1,1,1,1,5,5,-1,-1,1,1,-1,-1,-5,-5,-1,-1,1,1)^T% do all swaps (AB)
\end{multline*}

\item The final (reducible) subspace $S^{(3,1,1)^{\oplus^2}}$ is the span of twelve very basic vectors.
\begin{multline*}
(1, -1, 0, 0, 0, 0, 0, 0, 0, 0, 0, 0, 0, 0, 0, 0, 0, 0, 0, 0, 0, 0, 0, 0)^T,\\
(0, 0, 1, -1, 0, 0, 0, 0, 0, 0, 0, 0, 0, 0, 0, 0, 0, 0, 0, 0, 0, 0, 0, 0)^T,\\
(0, 0, 0, 0, 1, -1, 0, 0, 0, 0, 0, 0, 0, 0, 0, 0, 0, 0, 0, 0, 0, 0, 0, 0)^T,\\
(0, 0, 0, 0, 0, 0, 1, -1, 0, 0, 0, 0, 0, 0, 0, 0, 0, 0, 0, 0, 0, 0, 0, 0)^T,\\
(0, 0, 0, 0, 0, 0, 0, 0, 1, -1, 0, 0, 0, 0, 0, 0, 0, 0, 0, 0, 0, 0, 0, 0)^T,\\
(0, 0, 0, 0, 0, 0, 0, 0, 0, 0, 1, -1, 0, 0, 0, 0, 0, 0, 0, 0, 0, 0, 0, 0)^T,\\
(0, 0, 0, 0, 0, 0, 0, 0, 0, 0, 0, 0, 1, -1, 0, 0, 0, 0, 0, 0, 0, 0, 0, 0)^T,\\
(0, 0, 0, 0, 0, 0, 0, 0, 0, 0, 0, 0, 0, 0, 1, -1, 0, 0, 0, 0, 0, 0, 0, 0)^T,\\
(0, 0, 0, 0, 0, 0, 0, 0, 0, 0, 0, 0, 0, 0, 0, 0, 1, -1, 0, 0, 0, 0, 0, 0)^T,\\
(0, 0, 0, 0, 0, 0, 0, 0, 0, 0, 0, 0, 0, 0, 0, 0, 0, 0, 1, -1, 0, 0, 0, 0)^T,\\
(0, 0, 0, 0, 0, 0, 0, 0, 0, 0, 0, 0, 0, 0, 0, 0, 0, 0, 0, 0, 1, -1, 0, 0)^T,\\
(0, 0, 0, 0, 0, 0, 0, 0, 0, 0, 0, 0, 0, 0, 0, 0, 0, 0, 0, 0, 0, 0, 1, -1)^T
\end{multline*}

\end{itemize}
\caption{List of invariant subspaces for $CO_5$ ballots}
\label{figure:co5decomp}
\end{figure}

Before we look at specifics of the voting rules, let's now use these spaces (and the ones from $n=4$) to propose some potential voting ideas, with brief commentary.  The reader does not have to agree with the proposals!  They are meant solely as food for thought.  We do think they are reasonable, and in most cases the decompositions merely confirm what one might have thought of anyway.

\newpage
\begin{proposal}\label{proposal:useful}
The following should be considered as useful data on any voting system on cyclic orders, points-based or not:
\begin{itemize}
\item The number of total voters in a profile should be recorded.  
\begin{itemize}
\item This corresponds to the trivial subspace, as usual.  One could choose to kill it or to use it when adding profiles, as in normal voting.
\end{itemize}
\item When $n$ is odd, since the number of simple transpositions needed to convert a cyclic order into another one is invariant, this number and its parity may be taken into account in evaluating a system. 
\begin{itemize}
\item For instance, one could ask for a sign subspace to go to zero.
\end{itemize}
\item There is a fundamental distinction between systems that consider there to be no difference between a cyclic order and its reversal, and those which distinguish between them.
\begin{itemize}
\item One doesn't need the subspaces which are generated by profile differentials of the type $\{1,-1,0,\ldots ,0\}$ supporting an order and denigrating its reversal to assert this, but their consistency certainly helps.
\end{itemize}
\item Reversal pairs should be considered to behave in some ways like candidates in ordinary voting theory, and procedures should be designed/analyzed with this in mind.
\begin{itemize}
\item The $S^{(2,2)}$, $S^{(3,2)}$, and $S^{(2,2,1)}$ spaces support this line of thought.
\end{itemize}
\end{itemize}
\end{proposal}

Here is a less evident proposal.  Recall the computation of the character of $S^{(3,1,1)^{\oplus^2}}$ in Figure \ref{figure:co5decomp} and observe the distinction between $S^{(3,2)}$ and $S^{(2,2,1)}$.
\begin{proposal}\label{proposal:possible}
There is meaningful voting information involved in the distinction between `seating' people in a given cyclic order and instead seating them $k$ `seats' apart (in the same order) for some $k>1$, when this is possible number-theoretically.

This is a generalization of reversal, which may be considered to be the case $k=n-1$.  For $n=4$ no others are possible, but for $n=5$ several of these reorderings are possible since $5$ is prime.  For example, note the distinction between $(ABCDE)$ and $(ACEBD)$ via the \emph{right} action of $(1254)$, which is $k=3$ in this setting.
\end{proposal}

See Example \ref{ex:co5closertoborda} for evidence in favor of this.  We strongly suspect that \emph{additional} ideas will come with further study of the subspaces -- see also Section \ref{section:further}.

\subsection{Voting rules}

Although we will give the most general form of the matrix for $s:CO_5\to CO_5$ shortly, one $24\times 24$ matrix will suffice!   But we can give a voting rule simply from $s$ itself, as we have now practiced several times.

\begin{example}\label{ex:co5notsoborda}
Let $s:CO_5\to CO_5$ be the scoring function which gives points as follows:
\begin{itemize}
\item It awards four points if the two cyclic orders are the same.
\item If they differ by a transposition, score three points.
\item If they differ by two transpositions, score two points
\item If by three, score one point.
\item The reversal (which differs by four transpositions) scores zero points.
\end{itemize}
To be concrete, we have scores like $s((ABCDE),(ABDCE))=3$, $s((ABCDE),(AEDCB))=0$, and $s((ABCDE),(ABECD))=2$.  The way in which it is defined makes it clear it is neutral, and one could perhaps consider this to be a `Borda-like' system.

Now, what can we say about this system?  Even without a matrix (but see Proposition \ref{prop:5matrix}) it is pretty easy, if tedious, to compute that \emph{none} of the vectors in the decomposition of $\mathbb{Q}CO_5$ are killed.  In fact, most are simply dilated by a factor of two.  

However, the twelve-dimensional space $S^{(3,1,1)^{\oplus^2}}$ has a more complicated combination where a profile vector supporting a cyclic order (and opposing its reversal) is sent to an outcome vector giving:
\begin{itemize}
\item Four points to the order itself,
\item Negative four points to its reversal,
\item Two points to one differing by a transposition,
\item Negative two points to \emph{their} reversals, and
\item Zero to everything else.
\end{itemize}
That's a lot more like what we should be used to seeing from more-studied voting systems!
\end{example}

We can create a system which \emph{only} performs the latter transformation by taking the hint to have a sum-zero set of weights, along with one more change.

\begin{example}\label{ex:co5closertoborda}
Let $s:CO_5\to CO_5$ be the scoring function which gives points as follows:
\begin{itemize}
\item It awards two points if the two cyclic orders are the same.
\item If they differ by a transposition, score one point.
\item The reversal of one differing by a transposition scores negative one point.
\item The reversal will score negative two points.
\item All other scores are two points.
\end{itemize}
All scoring is the same as before (minus two points each), except that now \\$s((ABCDE),(ACEBD))=0$, which otherwise would have been $-1$ if we had just subtracted two (and its reversal would have also garnered $-1$ before).  (Recall Proposal \ref{proposal:possible}.)

This system behaves \emph{exactly the same} as the one in Example \ref{ex:co5notsoborda} on the $12$-dimensional subspace, but completely ignores the information from the other subspaces!   Which, to be honest, is a lot like what the Borda Count does; it amplifies one very specific piece of information (profiles which emphasize one particular candidate, symmetrically) and intentionally ignores any other information.  This is an appropriate system to end our initial investigations on!
\end{example}

It is possible to construct lengthier justifications of all this, but the easiest way to do so is for us to finally provide a matrix that represents any such voting rule.
The justification that this works is simply matrix multiplication on each of the subspaces in question, though of course \emph{finding} this matrix required ingenuity like in Proposition \ref{prop:4matrix}.  Our hope is that some of the ideas in Section \ref{section:further} will help in providing general results of this type in the future.

\newpage
\begin{prop}\label{prop:5matrix}
Any neutral points-based voting rule on $CO_5$ must have a matrix of scoring function values of the following form:

$$\begin{pmatrix}
a & b & c & d & c & d & e & f & e & f & d & c & c & d & e & f & g & h & e & f & c & d & f & e\\
b & a & d & c & d & c & f & e & f & e & c & d & d & c & f & e & h & g & f & e & d & c & e & f\\
c & d & a & b & e & f & d & c & c & d & e & f & f & e & g & h & e & f & c & d & e & f & c & d\\
d & c & b & a & f & e & c & d & d & c & f & e & e & f & h & g & f & e & d & c & f & e & d & c\\
c & d & e & f & a & b & c & d & d & c & e & f & e & f & d & c & e & f & c & d & f & e & h & g\\
d & c & f & e & b & a & d & c & c & d & f & e & f & e & c & d & f & e & d & c & e & f & g & h\\
e & f & d & c & c & d & a & b & e & f & c & d & c & d & f & e & d & c & f & e & h & g & f & e\\
f & e & c & d & d & c & b & a & f & e & d & c & d & c & e & f & c & d & e & f & g & h & e & f\\
e & f & c & d & d & c & e & f & a & b & c & d & h & g & e & f & d & c & f & e & d & c & e & f\\
f & e & d & c & c & d & f & e & b & a & d & c & g & h & f & e & c & d & e & f & d & c & f & e\\
d & c & e & f & e & f & c & d & c & d & a & b & f & e & d & c & f & e & h & g & f & e & c & d\\
c & d & f & e & f & e & d & c & d & c & b & a & e & f & c & d & e & f & g & h & e & f & d & c\\
c & d & f & e & e & f & c & d & g & h & f & e & a & b & c & d & e & f & d & c & f & e & c & d\\
d & c & e & f & f & e & d & c & h & g & e & f & b & a & d & c & f & e & c & d & e & f & d & c\\
e & f & h & g & d & c & f & e & e & f & d & c & c & d & a & b & d & c & e & f & d & c & e & f\\
f & e & g & h & c & d & e & f & f & e & c & d & d & c & b & a & c & d & f & e & c & d & f & e\\
h & g & e & f & e & f & d & c & d & c & f & e & e & f & d & c & a & b & d & c & e & f & c & d\\
g & h & f & e & f & e & c & d & c & d & e & f & f & e & c & d & b & a & c & d & f & e & d & c\\
e & f & c & d & c & d & f & e & f & e & g & h & d & c & e & f & d & c & a & b & d & c & f & e\\
f & e & d & c & d & c & e & f & e & f & h & g & c & d & f & e & c & d & b & a & c & d & e & f\\
c & d & e & f & f & e & g & h & c & d & f & e & f & e & d & c & e & f & d & c & a & b & d & c\\
d & c & f & e & e & f & h & g & d & c & e & f & e & f & c & d & f & e & c & d & b & a & c & d\\
f & e & c & d & g & h & f & e & e & f & c & d & c & d & e & f & c & d & f & e & d & c & a & b\\
e & f & d & c & h & g & e & f & f & e & d & c & d & c & f & e & d & c & e & f & c & d & b & a
\end{pmatrix}$$

This transformation scales each irreducible subspace in Figure \ref{figure:co5decomp} as follows:
\begin{itemize}
\item The trivial subspace is scaled by $a + b + 5c + 5d + 5e + 5f + g + h$.
\item The sign subspace is scaled by $a + b - 5c - 5d + 5e + 5f - g - h$.
\item The five-dimensional subspaces $S^{(3,2)}$ and $S^{(2,2,1)}$ are scaled by\\ $a + b - c - d - e - f + g + h$ and $a + b + c + d - e - f - g - h$, respectively.
\item Any one of the vectors in the twelve-dimensional subspace $S^{(3,1,1)^{\oplus^2}}$ corresponding to a cyclic order $x$ is sent to a vector with
\begin{itemize}
\item $a-b$ points for $x$ (and $b-a$ points for its reversal),
\item $c-d$ points for the orders differing by one transposition (and $d-c$ for their reversals),
\item $e-f$ points for the orders differing by a $3$-cycle (and $f-e$ for their reversals, which differ by two disjoint transpositions),
\item $g-h$ points for the (unique!) order differing as in Proposal \ref{proposal:possible}, and $h-g$ points for its reversal.
\end{itemize}
\end{itemize}
\end{prop}

\section{Further work and acknowledgments}\label{section:further}

It should be clear that this paper is just scratching the surface of a larger research program in voting on cyclic orders under various ballots.  In particular, we have only dealt with points-based procedures based on the scoring functions.  While successfully \emph{completely} characterizing all systems with ballot space $CO_n$ seems unlikely, future possible directions include the following.

\begin{itemize}
\item Deciding which, if any, of the ideas in Proposals \ref{proposal:useful} or \ref{proposal:possible} are worthy of further investigation.
\item Hence, an important extension is to analyze voting when we do \emph{not} worry about reversal.  This halves the number of cyclic orders, and one may note that the decompositions already provided in this paper clearly have some subspaces which would remain invariant if reversed orders were identified.  There are evident generalizations of Theorem \ref{thm:cocharacter} which would immediately yield further detail.
\item Unlike some alternate ballots, the ROLO systems are easily extended to $n\geq 4$.  As noted above, the situation of more ballots than cyclic orders is a low-dimensional phenomenon.
\item Finally, there are important graphs that are associated with the various sets of cyclic orders (see also below).  They have additional symmetry which may lead to a finer-grained analysis.  
\end{itemize}

In addition (or perhaps more importantly), using the algebra to come up with axiomatics of the usual voting variety would be a very good goal.  It seems very unlikely that there are not analogues to the Borda or Kemeny rules (for instance, as maximum likelihood estimators).  Alternately to raw axiomatics, using the methods of \cite{BarthelemyMonjardetMedian} or \cite{Zwicker} might be a good place to start.  This would be fairly necessary to address ballots consisting of full or truncated rankings of cyclic orders, analysis of which otherwise would look pretty much like similar ballots in `ordinary' voting theory.

A different direction for further research is the set of cyclic orders\footnote{Also called circular orders, circular permutations, etc.  These objects have often been rediscovered.}.  There is quite a bit of work on cyclic orders considered as ternary relations, not focusing at all on overall structure of the set of cyclic orders.  For example, one may ask whether a partial cyclic order (as a relation) can be completed to a total cyclic order (like the ones under discussion in this paper); this problem is actually NP-complete.  See among many others \cite{Megiddo,Haar,Quilliot} and especially many papers by Nov\'ak, such as \cite{NovakNovotnyCyclic,NovotnyNovakCyclic}.
They also appear in many graph-theoretic contexts, such as graph crossing numbers; the paper \cite{WoodallCyclicOrder} proves some very interesting things about a graph\footnote{For graphs, see also \cite{DeTempleRobertson}, \cite{NahasCrossing}, and \cite{KleitmanCrossing}.} with cyclic orders as vertices.  Nonetheless there is a lot more which could be asked about them\footnote{Cyclic orders are also used in defining so-called ribbon graphs in mathematical physics, but there they simply fulfill the role of a topological datum.}, particularly regarding distances on the set.  We welcome any references we are not aware of for this interesting object.

The authors wish to particularly thank the Gordon College Provost's Summer Undergraduate Research Fellowship for support, and the Gordon College Department of Mathematics and Computer Science, Mike Orrison, and Bill Zwicker for helpful early feedback.  We also wish to thank the organizers of the AMS Special Session on The Mathematics of Decisions, Elections and Games for inviting a talk based on this research, for organizing a volume of papers, and for helpful comments from the editors and a referee.  We finally want to recognize the teams behind Sage Math (\cite{Sage}) and GAP (\cite{GAP4}) for software that greatly eased the exploration of these topics, particularly with respect to computing characters and kernels.

\bibliographystyle{amsplain}
\bibliography{refs}

\end{document}